\documentclass[11pt]{article}
\usepackage[a4paper,margin=26mm]{geometry}
\usepackage[T1]{fontenc}
\usepackage[utf8]{inputenc}
\usepackage{lmodern,amsmath,amssymb,amsthm,mathtools,microtype}
\usepackage{booktabs,array,enumitem,tikz}
\usepackage[colorlinks=true,linkcolor=blue,citecolor=blue,urlcolor=blue]{hyperref}
\hypersetup{pdftitle={Polynomial-time local-unitary equivalence of graph states},pdfauthor={Yuxuan Zhang}}
\newtheorem{theorem}{Theorem}[section]
\newtheorem{lemma}[theorem]{Lemma}
\newtheorem{proposition}[theorem]{Proposition}
\newtheorem{corollary}[theorem]{Corollary}
\theoremstyle{definition}
\newtheorem{definition}[theorem]{Definition}
\theoremstyle{remark}

\newcommand{\F}{\mathbb F}
\newcommand{\Z}{\mathbb Z}
\newcommand{\C}{\mathbb C}
\newcommand{\LC}{\mathrm{LC}}

\DeclareMathOperator{\supp}{supp}
\DeclareMathOperator{\spanR}{span}
\DeclareMathOperator{\Odd}{Odd}
\DeclareMathOperator{\diag}{diag}
\setlist[enumerate]{itemsep=3pt,topsep=5pt}
\title{Polynomial-time local-unitary equivalence\\of graph states}
\author{Yuxuan Zhang\thanks{Correspondence: \href{mailto:yuxuan.zhang@epfl.ch}{\texttt{yuxuan.zhang@epfl.ch}}.}\\[3pt]
\small Institute of Physics, \'Ecole Polytechnique F\'ed\'erale de Lausanne (EPFL)\\
\small CH-1015 Lausanne, Switzerland\\
\small Department of Physics, Princeton University\\
\small Princeton, New Jersey 08544, USA}
\date{}
\begin{document}
\maketitle

\begin{abstract}
Local-unitary (LU) equivalence asks whether two quantum states differ only by independent changes of basis on their qubits. For graph states, whether this relation can be decided in polynomial time has remained open for over a decade. We give a deterministic algorithm that decides LU equivalence for graphs on $n$ labelled vertices in $\widetilde O(n^{6.38})$ bit operations and constructs exact single-qubit unitaries whenever the states are equivalent. Building on Claudet and Perdrix's quasipolynomial algorithm, we replace the enumeration of vertex subsets by a compact system of constraints generated from pairs and triples. The remaining graph transformation is found by solving linear equations over the binary field. These new steps cost $\widetilde O(n^5)$ bit operations; the inherited graph preprocessing sets the overall bound. We also count the local-Clifford (LC) classes of graph states within any LU class: their number is a power of two, computable within the same bound. For any given graph state, this decides whether single-qubit Clifford gates reach every graph state in its LU class, and supplies a counterexample when they do not. The method also decides LU equivalence of stabilizer codes encoding one logical qubit.
\end{abstract}

\section{Introduction}\label{sec:notation}

The graph state required by a computation or communication protocol may be difficult to prepare on the available device. An alternative is to prepare a different graph state and obtain the desired resource using single-qubit gates. Such local changes of basis can alter the edge set while preserving multipartite entanglement, giving us a choice of preparation graphs. To use this freedom, we need to determine whether two given graphs are related in this way and find the local gates. Whether this local-unitary (LU) equivalence problem can be solved in polynomial time has remained open for over a decade \cite[Section~5]{Ji2010}. We give a constructive polynomial-time solution.

For a finite simple graph $G=(V,E)$ on $V=[n]=\{1,\ldots,n\}$, prepare each qubit in $|+\rangle$ and apply a controlled-$Z$ gate on each edge. The resulting graph state is
\begin{equation}\label{eq:graphstate}
 |G\rangle=\prod_{\{u,v\}\in E}\mathrm{CZ}_{uv}|+\rangle^{\otimes n},
 \qquad |+\rangle=\frac{|0\rangle+|1\rangle}{\sqrt2},
\end{equation}
where $\mathrm{CZ}_{uv}$ acts on qubits $u,v$. Given the binary adjacency matrices of $G,H$, we ask whether
\begin{equation}\label{eq:LU}
 |H\rangle=e^{i\phi}\left(\bigotimes_{v=1}^n U_v\right)|G\rangle
 \quad\text{for some }U_v\in U(2),\ \phi\in\mathbb R.
\end{equation}
The vertex labels fix the correspondence between qubits; permutations are not allowed.

The graph description uses only $O(n^2)$ classical bits to specify multipartite entanglement \cite{BRCluster,HeinGraph}. It covers every pure qubit stabilizer state up to single-qubit Clifford operations \cite{VDDGraph} and also gives constructions of quantum error-correcting codes \cite{SWCodes}. In measurement-based quantum computation (MBQC), the graph state is the resource prepared in advance; the computation then proceeds through local measurements and classical feedforward \cite{RB,RBB}. Cluster states support universal computation and topological fault tolerance \cite{RHG}. A change of local basis in the resource can be absorbed into the measurement choices whenever the required bases are available.

How that resource is prepared depends on the hardware. Graph-state preparation and measurement-based computation have been demonstrated with photons \cite{Walther,LuGraph} and trapped ions \cite{Lanyon}, and neutral-atom arrays allow programmable graph-state preparation \cite{Bluvstein}. Photonic resources can be assembled by fusion operations, which also underlie fusion-based computing architectures \cite{BrowneRudolph,Bartolucci}. Sequential emission from a quantum emitter provides another route \cite{Lindner}, realized with a quantum dot \cite{Schwartz} and a single atom \cite{Thomas}.\footnote{Continuous-variable cluster states have also been generated by temporal multiplexing \cite{Asavanant,Larsen}; photonic quantum-volume benchmarks use such resources with Gottesman--Kitaev--Preskill encoding \cite{ZhangVolume}. The present equivalence theorem concerns qubit graph states.}

The same resource-state approach appears in quantum repeaters. Purification-based repeater schemes \cite{BriegelRepeaters} have measurement-based versions in which graph-state measurements implement purification and entanglement swapping \cite{Zwerger}. Loss-tolerant cluster measurements \cite{Varnava} enable all-photonic repeater proposals \cite{Azuma}; experimental steps include adaptive Bell measurements \cite{Hasegawa} and small all-photonic repeater circuits \cite{LiRepeater}. Graph-state realizations of Calderbank--Shor--Steane codes also support one-way repeaters with measurement-based error correction \cite{NiuRepeaters}.

For Clifford changes of basis, equivalence already has a polynomial-time test. Single-qubit Clifford unitaries map Pauli operators to Pauli operators under conjugation. Equivalence under their tensor products, denoted $\LC$, is captured exactly by a graph operation: two graph states are LC-equivalent if and only if their graphs are related by local complementations \cite{VDDGraph}. A local complementation at $v$ toggles the edges among distinct neighbors of $v$; write $G\star v$ for the resulting graph. Bouchet's algorithm and its stabilizer formulation decide this relation in polynomial time \cite{Bouchet,VDD}.\footnote{Allowing local Pauli measurements as well leads to a broader resource-conversion problem, formulated through graph vertex-minors by Dahlberg and Wehner \cite{Dahlberg}.}

The LU--LC conjecture would have made this test sufficient: it asserted that any two LU-equivalent graph states are also LC-equivalent. Van den Nest, Dehaene, and De Moor established sufficient conditions for the conjecture to hold \cite{VDDCriterion}, later extended by Zeng and collaborators \cite{Zeng}. Gross and Van den Nest related it to diagonal local operations and quadratic forms \cite{GrossVDN}. The general conjecture failed when Ji, Chen, Wei, and Ying constructed a 27-qubit counterexample \cite{Ji2010}. Further examples followed \cite{TG}, and Claudet proved that 27 qubits are necessary \cite[Theorem~1]{Minimal27}. Recognizing LU equivalence therefore requires a way to account for transformations that ordinary local complementation misses.

Burchardt, de Jong, and Vandr\'e formulated LU verification as modular linear equations \cite{BDV}, although their system can contain exponentially many equations and its general complexity remained unresolved \cite[Sections~V and~VII]{BDV}. Claudet and Perdrix supplied a graphical characterization \cite{Graphical} and a quasipolynomial algorithm, running in time $n^{\log_2 n+O(1)}$ \cite{CP}. They put the input graphs in a common standard form, compute the edge changes obtainable by a generalized local complementation, and finish with a constrained LC test. Their enumeration of incidence congruences over vertex subsets of logarithmically growing size is the costly step. Those congruences have a recursive structure: a small set of rows generates all the others. Computing this closure gives the polynomial bound below and answers the question also posed in Claudet's thesis \cite[Chapter~8]{Claudet}.

\begin{theorem}[Main result]\label{thm:main}
There is a deterministic classical algorithm that, given two finite simple graphs $G,H$ on the same labelled vertex set of size $n\geq1$, decides in $\widetilde O(n^{6.38})$ bit operations whether their graph states are local-unitary equivalent. If they are, it returns single-qubit unitaries satisfying \eqref{eq:LU}, with an exact symbolic description of $O(n\log(n+1))$ bits.
\end{theorem}

Here $\widetilde O$ suppresses logarithmic factors. We retain Claudet and Perdrix's standard form, their reduction to one generalized complementation, and their logarithmic level bound. The incidence closure and the affine graph equations developed here cost $\widetilde O(n^5)$ bit operations after preprocessing. The exponent $6.38$ comes from the inherited common-standard-form procedure.

In standard form, the generalized complementation acts on an independent set $X$. Put $J=V\setminus X$ and work modulo $q=2^r$, where $r=O(\log n)$. For $v\in J$, let $a_v\in\{0,1\}^X$ record its neighbors in $X$, and let $a_K$ be the coordinatewise product over $v\in K$. The incidence row for a pair or triple is $2a_K$; for a larger set it is $2^{|K|-2}a_K$. Adding a vertex to the index set amounts to multiplying coordinatewise by its neighborhood vector and doubling the row. For $|K|\geq3$ and $v\notin K$,
\[
 2a_v\odot\bigl(2^{|K|-2}a_K\bigr)
       =2^{|K|-1}a_{K\cup\{v\}},
\]
where $\odot$ denotes coordinatewise multiplication. Every higher-order row can therefore be reached from a triple. Pairs must be supplied separately, and triples must also be included at the start: applying the map to a pair gives twice the required triple coefficient. Since the maps are linear, they can be applied to generators of the row module, with redundant rows removed after each round. Modular elimination keeps at most $|X|$ generators with $O(r)$ bits per entry, and $O(r)$ rounds suffice (Theorem~\ref{thm:compression}).

The compressed system gives a binary vector space of allowed edge changes. We still need to choose one change and the ordinary complementations that take the resulting graph to the target. Write $Z$ for the vertices whose mutual edges can change, and $W$ for the permitted complementation centers. Restricting the centers to $W$ fixes the relevant Clifford coordinates on $X\cup Z$. The remaining graph equations are then affine in both the edge-change coefficients and the unknown Clifford coordinates. A determinant identity propagates invertibility from one fixed vertex across a connected target, so solving these equations gives the edge change and the remaining local Cliffords together (Theorem~\ref{thm:direct}). All equations use the original vertices. Sections~\ref{sec:image} and~\ref{sec:direct} explain the relation to earlier modular elimination and constant-determinant results.

Keeping the multiplicity vectors and the Clifford operations throughout the calculation recovers the exact unitary in \eqref{eq:LU}. The affine solution space also describes every feasible edge change (Corollary~\ref{cor:feasible}). Reusing the standard form while varying the level determines the least dyadic level needed for the transformation (Corollary~\ref{cor:minlevel}).

A single standard form also suffices to represent every LC class within the starting graph's LU class. The allowed edge changes form one binary vector space; those that ordinary LC transformations can undo form another. Their quotient indexes the LC classes, giving a power-of-two count computable in polynomial time (Theorem~\ref{thm:orbitcount}). If the quotient is nontrivial, the algorithm constructs an LU-equivalent but LC-inequivalent graph and the connecting unitary. Counting individual labelled graphs has a different complexity: combined with the hardness result of Hahn et al., the decision algorithm gives $\#\mathsf P$-completeness (Corollary~\ref{cor:graphcount}).

Graph-state representatives can be obtained efficiently from Pauli generators, so the test applies to pure stabilizer states as well. For these states it also decides equivalence under stochastic local operations and classical communication (SLOCC). Purifying a stabilizer code encoding one logical qubit reduces code-space LU equivalence to the same test (Corollaries~\ref{cor:slocc} and~\ref{cor:onequbitcode}).

Section~\ref{sec:reduction} recalls the Claudet--Perdrix reduction and states the incidence constraints. Section~\ref{sec:compression} computes the admissible edge changes by incidence compression, kernel recovery, and binary image extraction. Section~\ref{sec:integration} solves the remaining graph transformation and assembles the constructive equivalence algorithm, including its consequences for stabilizer states and codes. Section~\ref{sec:orbitcount} proves the LC-class classification and counting results, and Section~\ref{sec:discussion} discusses applications and further questions. Appendices~\ref{app:reduction}--\ref{app:clifford} give the connected-component argument, the standardization cost, computational checks, and extraction of a local-complementation sequence.

\section{The Claudet--Perdrix reduction}\label{sec:reduction}\label{sec:dependencies}

\subsection{Standard form}\label{sec:standardform}

Write $\Gamma(G)$ for the adjacency matrix of $G$ over the binary field $\F_2$, and let $\prec$ be the vertex order inherited from $[n]$. With the graph fixed, $\Gamma_{A,B}$ denotes the block with rows in $A$ and columns in $B$. The neighborhood of $u$ is $N_G(u)$; for a vertex set $A$, put $N_G(A)=\bigcup_{u\in A}N_G(u)$. LU equivalence preserves the labelled connected-component partition (Lemma~\ref{lem:components}), so we reject different partitions and then work componentwise. We use $n$ for the full input order and $n_c$ for the component order. One-vertex components require no transformation; the preprocessing below applies when $n_c\geq2$.

The standard form is defined through the supports of graph-state stabilizers. For $D\subseteq V$, put
\[
 \Odd_G(D)=\{u\in V:|N_G(u)\cap D|\equiv1\pmod2\}.
\]
A \emph{local set} is a nonempty set $L=D\cup\Odd_G(D)$ for some $D\subseteq V$; such a $D$ is called a generator of $L$. A local set is minimal if it contains no proper nonempty local set. A \emph{minimal-local-set cover}, or MLS cover, is a family $\mathcal M$ of minimal local sets whose union is $V$.

Local sets are precisely the nonempty supports of graph-state stabilizers. Write $\mathsf Z_A=\prod_{v\in A}\mathsf Z_v$ for a Pauli product on a vertex set $A$. The stabilizer is generated by $\mathsf X_u\mathsf Z_{N_G(u)}$. If $\mathbf d$ is the indicator of $D$, the product indexed by $D$ has support $D\cup\Odd_G(D)$, and its Pauli at $u$, up to phase, is specified by $(\mathbf d_u,(\Gamma(G)\mathbf d)_u)$. The pairs $(1,0),(1,1),(0,1)$ represent $\mathsf X,\mathsf Y,\mathsf Z$, respectively. Relative to a fixed cover $\mathcal M$, a vertex has type $\mathsf X$, $\mathsf Y$, or $\mathsf Z$ if every generator of every member of $\mathcal M$ containing it gives that same nonzero Pauli pair; otherwise it has type $\bot$ \cite[Definitions~12--13]{CP}.\footnote{We follow the numbering of the PDF of arXiv:2502.06566v3.} Thus a type records a Pauli axis shared by the minimal-support stabilizers in the chosen cover. Type $\bot$ means that these stabilizers impose more than one Pauli axis at the vertex.

A graph is in \emph{standard form} with respect to $\mathcal M$ if it has no type-$\mathsf Y$ vertices; every neighbor $v$ of a type-$\mathsf X$ vertex $u$ has type $\mathsf Z$ and satisfies $u\prec v$; and each star $\{u\}\cup N_G(u)$ at a type-$\mathsf X$ vertex belongs to $\mathcal M$ \cite[Definition~14]{CP}. Write $X,Z,W$ for the sets of type-$\mathsf X$, type-$\mathsf Z$, and type-$\bot$ vertices. Then $V=X\sqcup Z\sqcup W$, $X$ is independent, and $N_G(X)\subseteq Z$. Figure~\ref{fig:standardform} shows the resulting geometry.

The common-standard-form algorithm either certifies that the two input graphs are not LU-equivalent or returns ordinary-LC-equivalent representatives $G_1,G_2$ in standard form with respect to the same MLS cover, in $O(n_c^{6.38})$ time \cite[Lemma~16]{CP}. The cover is computed in polynomial time without enumerating all minimal local sets \cite{MLS}. Ordinary local complementation preserves connectedness, since any edge it removes has a replacement path through the complemented vertex. Thus $G_1$ and $G_2$ remain connected.

LU-equivalent standard forms have identical vertex types and identical neighborhoods at type-$\mathsf X$ vertices \cite[Lemma~17]{CP}. We check these two conditions and reject a pair if either fails. A pair that passes is called a \emph{surviving pair}; it has a common partition $X,Z,W$ and $N_{G_1}(x)=N_{G_2}(x)$ for every $x\in X$.

\subsection{Incidence congruences}\label{sec:incidence}

Claudet and Perdrix reduce the continuous freedom in local unitaries to a hierarchy of gate sets with dyadic phase angles. For an integer $r\geq1$, $\LC_r$-equivalence allows at each qubit products of Hadamard gates and $\diag(1,e^{i\pi/2^r})$. The integer $r$ is called the \emph{level}; level one gives local Cliffords, and level two gives local Clifford+$T$ gates, with $T=\diag(1,e^{i\pi/4})$. Every $\LC_r$ transformation is local unitary.

\begin{theorem}[Claudet--Perdrix {\cite[Corollary~29]{CP}}]\label{thm:level}
Let $r\geq1$. If two graphs on the same labelled vertex set of size at most $2^{r+3}-1$ are LU-equivalent, they are $\LC_r$-equivalent.
\end{theorem}

For a component of order $n_c\geq2$, choose
\begin{equation}\label{eq:level}
 r=\max\{1,\lceil\log_2((n_c+1)/8)\rceil\}.
\end{equation}
This is the smallest integer $r\geq1$ with $n_c\leq2^{r+3}-1$, so LU- and $\LC_r$-equivalence coincide on the component. In particular, $r=O(\log n_c)$ and $q=2^r=O(n_c)$; components with $n_c\leq15$ use $r=1$.

For a multiset $S$ of vertices, $s(x)$ is its nonnegative integer multiplicity and $\supp(S)=\{x:s(x)>0\}$. The multiset is \emph{independent} when its support is an independent vertex set. For $K\subseteq V$, define the common-neighbor multiplicity
\begin{equation}\label{eq:common}
 c_S(K)=\sum_{x:\ K\subseteq N_G(x)}s(x).
\end{equation}

\begin{definition}[{\cite[Definitions~1--2]{CP}}]\label{def:incident}
Let $r\geq1$. An independent multiset $S$ is $r$-incident if, for every $0\leq j<r$ and every $K\subseteq V\setminus\supp(S)$ with $|K|=j+2$,
\begin{equation}\label{eq:incidence}
 c_S(K)\equiv0\pmod{2^{r-j-\delta_{j0}}},
\end{equation}
where $\delta_{j0}=1$ for $j=0$ and is zero otherwise. Its $r$-local complementation, denoted $G\star_r S$, toggles a pair outside $\supp(S)$ precisely when its common-neighbor multiplicity is $2^{r-1}$ modulo $2^r$. Pairs meeting the support are unchanged.
\end{definition}

The graph rule is implemented by rotations about local Pauli axes. Put $q=2^r$, $Z(\theta)=\diag(1,e^{i\theta})$, and $X(\theta)=\mathsf H Z(\theta)\mathsf H$, where $\mathsf H$ is the single-qubit Hadamard gate. For an $r$-incident multiset with multiplicities $s$ supported on an independent set $X$, Claudet and Perdrix give \cite[Section~2.3]{CP}
\begin{equation}\label{eq:generalizedunitary}
 |G\star_r s\rangle=
 \left(\bigotimes_{u\in V}X\!\left(\frac{s(u)\pi}{q}\right)\right)
 \left(\bigotimes_{v\in V}Z\!\left(-\frac{\pi}{q}
           \sum_{u\in N_G(v)}s(u)\right)\right)|G\rangle,
\end{equation}
where $s$ is zero outside $X$. Independence ensures that a vertex in $X$ receives only an $X$-rotation; vertices outside $X$ receive compensating $Z$-rotations. The incidence conditions ensure that the transformed state is again a graph state. In standard form, only edges within $Z$ can change.

\begin{figure}[tbp]
\centering
\begin{minipage}[t]{0.48\linewidth}
\centering
\textbf{(a)} Standard form\par\smallskip
\begin{tikzpicture}[x=1cm,y=1cm,font=\footnotesize,
 xqubit/.style={circle,draw=blue!65!black,fill=blue!12,minimum size=5mm,inner sep=1pt},
 zqubit/.style={circle,draw=orange!70!black,fill=orange!12,minimum size=5mm,inner sep=1pt},
 wqubit/.style={circle,draw,fill=gray!12,minimum size=5mm,inner sep=1pt}]
\path[use as bounding box] (-0.7,-0.9) rectangle (5.2,3.9);
\node at (0,3.45) {$X$}; \node at (2.1,3.45) {$Z$}; \node at (4.2,3.45) {$W$};
\foreach \i/\y in {1/2.8,2/1.6,3/0.4} {
 \node[xqubit] (x\i) at (0,\y) {$x_{\i}$};
 \node[zqubit] (z\i) at (2.1,\y) {$z_{\i}$};
}
\node[wqubit] (w1) at (4.2,2.3) {$w_1$};
\node[wqubit] (w2) at (4.2,0.9) {$w_2$};
\foreach \i/\j in {1/1,1/2,2/1,2/3,3/2,3/3} \draw[gray] (x\i)--(z\j);
\draw[gray] (z1)--(w1)--(w2)--(z3); \draw[gray] (z2)--(w1);
\draw[orange!80!black,dashed,thick] (z1) to[bend left=35] (z2);
\draw[orange!80!black,dashed,thick] (z2) to[bend left=35] (z3);
\draw[orange!80!black,dashed,thick] (z1) to[bend left=62] (z3);
\node[align=center] at (2.1,-0.45) {$X$ independent; $N(X)\subseteq Z$\\generalized step changes only $Z$--$Z$ edges};
\end{tikzpicture}
\end{minipage}\hfill
\begin{minipage}[t]{0.48\linewidth}
\centering
\textbf{(b)} A 27-qubit LU-but-not-LC pair\par\smallskip
\begin{tikzpicture}[x=1cm,y=1cm,font=\footnotesize]
\path[use as bounding box] (-2.7,-2.7) rectangle (2.7,2.1);
\coordinate (x1) at (90.00000:1.95);
\coordinate (x2) at (107.14286:1.95);
\coordinate (x3) at (124.28571:1.95);
\coordinate (x4) at (141.42857:1.95);
\coordinate (x5) at (158.57143:1.95);
\coordinate (x6) at (175.71429:1.95);
\coordinate (x7) at (192.85714:1.95);
\coordinate (x8) at (210.00000:1.95);
\coordinate (x9) at (227.14286:1.95);
\coordinate (x10) at (244.28571:1.95);
\coordinate (x11) at (261.42857:1.95);
\coordinate (x12) at (278.57143:1.95);
\coordinate (x13) at (295.71429:1.95);
\coordinate (x14) at (312.85714:1.95);
\coordinate (x15) at (330.00000:1.95);
\coordinate (x16) at (347.14286:1.95);
\coordinate (x17) at (364.28571:1.95);
\coordinate (x18) at (381.42857:1.95);
\coordinate (x19) at (398.57143:1.95);
\coordinate (x20) at (415.71429:1.95);
\coordinate (x21) at (432.85714:1.95);
\coordinate (z1) at (90:0.78);
\coordinate (z2) at (150:0.78);
\coordinate (z3) at (210:0.78);
\coordinate (z4) at (270:0.78);
\coordinate (z5) at (330:0.78);
\coordinate (z6) at (390:0.78);
\draw[gray!45,line width=0.25pt] (x1)--(z1);
\draw[gray!45,line width=0.25pt] (x1)--(z2);
\draw[gray!45,line width=0.25pt] (x1)--(z3);
\draw[gray!45,line width=0.25pt] (x1)--(z4);
\draw[gray!45,line width=0.25pt] (x2)--(z1);
\draw[gray!45,line width=0.25pt] (x2)--(z2);
\draw[gray!45,line width=0.25pt] (x2)--(z3);
\draw[gray!45,line width=0.25pt] (x2)--(z5);
\draw[gray!45,line width=0.25pt] (x3)--(z1);
\draw[gray!45,line width=0.25pt] (x3)--(z2);
\draw[gray!45,line width=0.25pt] (x3)--(z3);
\draw[gray!45,line width=0.25pt] (x3)--(z6);
\draw[gray!45,line width=0.25pt] (x4)--(z1);
\draw[gray!45,line width=0.25pt] (x4)--(z2);
\draw[gray!45,line width=0.25pt] (x4)--(z4);
\draw[gray!45,line width=0.25pt] (x4)--(z5);
\draw[gray!45,line width=0.25pt] (x5)--(z1);
\draw[gray!45,line width=0.25pt] (x5)--(z2);
\draw[gray!45,line width=0.25pt] (x5)--(z4);
\draw[gray!45,line width=0.25pt] (x5)--(z6);
\draw[gray!45,line width=0.25pt] (x6)--(z1);
\draw[gray!45,line width=0.25pt] (x6)--(z2);
\draw[gray!45,line width=0.25pt] (x6)--(z5);
\draw[gray!45,line width=0.25pt] (x6)--(z6);
\draw[gray!45,line width=0.25pt] (x7)--(z1);
\draw[gray!45,line width=0.25pt] (x7)--(z3);
\draw[gray!45,line width=0.25pt] (x7)--(z4);
\draw[gray!45,line width=0.25pt] (x7)--(z5);
\draw[gray!45,line width=0.25pt] (x8)--(z1);
\draw[gray!45,line width=0.25pt] (x8)--(z3);
\draw[gray!45,line width=0.25pt] (x8)--(z4);
\draw[gray!45,line width=0.25pt] (x8)--(z6);
\draw[gray!45,line width=0.25pt] (x9)--(z1);
\draw[gray!45,line width=0.25pt] (x9)--(z3);
\draw[gray!45,line width=0.25pt] (x9)--(z5);
\draw[gray!45,line width=0.25pt] (x9)--(z6);
\draw[gray!45,line width=0.25pt] (x10)--(z1);
\draw[gray!45,line width=0.25pt] (x10)--(z4);
\draw[gray!45,line width=0.25pt] (x10)--(z5);
\draw[gray!45,line width=0.25pt] (x10)--(z6);
\draw[gray!45,line width=0.25pt] (x11)--(z2);
\draw[gray!45,line width=0.25pt] (x11)--(z3);
\draw[gray!45,line width=0.25pt] (x11)--(z4);
\draw[gray!45,line width=0.25pt] (x11)--(z5);
\draw[gray!45,line width=0.25pt] (x12)--(z2);
\draw[gray!45,line width=0.25pt] (x12)--(z3);
\draw[gray!45,line width=0.25pt] (x12)--(z4);
\draw[gray!45,line width=0.25pt] (x12)--(z6);
\draw[gray!45,line width=0.25pt] (x13)--(z2);
\draw[gray!45,line width=0.25pt] (x13)--(z3);
\draw[gray!45,line width=0.25pt] (x13)--(z5);
\draw[gray!45,line width=0.25pt] (x13)--(z6);
\draw[gray!45,line width=0.25pt] (x14)--(z2);
\draw[gray!45,line width=0.25pt] (x14)--(z4);
\draw[gray!45,line width=0.25pt] (x14)--(z5);
\draw[gray!45,line width=0.25pt] (x14)--(z6);
\draw[gray!45,line width=0.25pt] (x15)--(z3);
\draw[gray!45,line width=0.25pt] (x15)--(z4);
\draw[gray!45,line width=0.25pt] (x15)--(z5);
\draw[gray!45,line width=0.25pt] (x15)--(z6);
\draw[gray!45,line width=0.25pt] (x16)--(z1);
\draw[gray!45,line width=0.25pt] (x16)--(z2);
\draw[gray!45,line width=0.25pt] (x16)--(z3);
\draw[gray!45,line width=0.25pt] (x16)--(z4);
\draw[gray!45,line width=0.25pt] (x16)--(z5);
\draw[gray!45,line width=0.25pt] (x17)--(z1);
\draw[gray!45,line width=0.25pt] (x17)--(z2);
\draw[gray!45,line width=0.25pt] (x17)--(z3);
\draw[gray!45,line width=0.25pt] (x17)--(z4);
\draw[gray!45,line width=0.25pt] (x17)--(z6);
\draw[gray!45,line width=0.25pt] (x18)--(z1);
\draw[gray!45,line width=0.25pt] (x18)--(z2);
\draw[gray!45,line width=0.25pt] (x18)--(z3);
\draw[gray!45,line width=0.25pt] (x18)--(z5);
\draw[gray!45,line width=0.25pt] (x18)--(z6);
\draw[gray!45,line width=0.25pt] (x19)--(z1);
\draw[gray!45,line width=0.25pt] (x19)--(z2);
\draw[gray!45,line width=0.25pt] (x19)--(z4);
\draw[gray!45,line width=0.25pt] (x19)--(z5);
\draw[gray!45,line width=0.25pt] (x19)--(z6);
\draw[gray!45,line width=0.25pt] (x20)--(z1);
\draw[gray!45,line width=0.25pt] (x20)--(z3);
\draw[gray!45,line width=0.25pt] (x20)--(z4);
\draw[gray!45,line width=0.25pt] (x20)--(z5);
\draw[gray!45,line width=0.25pt] (x20)--(z6);
\draw[gray!45,line width=0.25pt] (x21)--(z2);
\draw[gray!45,line width=0.25pt] (x21)--(z3);
\draw[gray!45,line width=0.25pt] (x21)--(z4);
\draw[gray!45,line width=0.25pt] (x21)--(z5);
\draw[gray!45,line width=0.25pt] (x21)--(z6);
\draw[orange!85!black,dashed,line width=0.6pt] (z1)--(z2);
\draw[orange!85!black,dashed,line width=0.6pt] (z1)--(z3);
\draw[orange!85!black,dashed,line width=0.6pt] (z1)--(z4);
\draw[orange!85!black,dashed,line width=0.6pt] (z1)--(z5);
\draw[orange!85!black,dashed,line width=0.6pt] (z1)--(z6);
\draw[orange!85!black,dashed,line width=0.6pt] (z2)--(z3);
\draw[orange!85!black,dashed,line width=0.6pt] (z2)--(z4);
\draw[orange!85!black,dashed,line width=0.6pt] (z2)--(z5);
\draw[orange!85!black,dashed,line width=0.6pt] (z2)--(z6);
\draw[orange!85!black,dashed,line width=0.6pt] (z3)--(z4);
\draw[orange!85!black,dashed,line width=0.6pt] (z3)--(z5);
\draw[orange!85!black,dashed,line width=0.6pt] (z3)--(z6);
\draw[orange!85!black,dashed,line width=0.6pt] (z4)--(z5);
\draw[orange!85!black,dashed,line width=0.6pt] (z4)--(z6);
\draw[orange!85!black,dashed,line width=0.6pt] (z5)--(z6);
\foreach \i in {1,...,21} \filldraw[draw=blue!65!black,fill=blue!18] (x\i) circle (2pt);
\foreach \i in {1,...,6} \filldraw[draw=orange!80!black,fill=white,line width=0.8pt] (z\i) circle (3pt);
\node[align=center] at (0,-2.45) {outer ring: $|X|=21$; inner ring: $|Z|=6$\\dashed edges: the added clique on $Z$};
\end{tikzpicture}
\end{minipage}
\caption{Generalized local complementation in standard form. (a) The schematic shows the permitted adjacencies; dashed edges mark the pairs that a generalized step on $X$ may toggle. (b) In the Tsimakuridze--G\"uhne example \cite{TG}, each outer vertex is adjacent to one four- or five-element subset of the six inner vertices. The first graph contains only the gray incidence edges; the second also contains the dashed clique. A level-two local unitary connects their graph states, although no local Clifford does.}
\label{fig:standardform}
\end{figure}
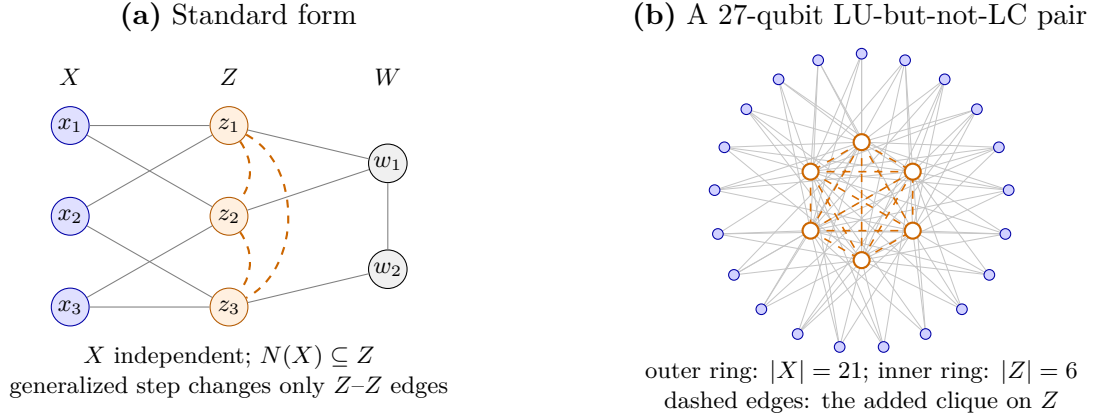

\begin{theorem}[Claudet--Perdrix {\cite[Lemma~18]{CP}}]\label{thm:standardreduction}
Let $r\geq1$, and let $G_1,G_2$ be $\LC_r$-equivalent graphs in standard form with respect to the same MLS cover. They are related by a sequence of ordinary local complementations at type-$\bot$ vertices together with a single valid $r$-local complementation supported on the type-$\mathsf X$ vertices.
\end{theorem}

The statements used here and in the fixed-cover argument of Section~\ref{sec:orbitcount} apply to an arbitrary chosen MLS cover. Detailed proofs are given in \cite[Lemmas~10--11 and~14]{Claudet}, with the local-axis restrictions established in \cite[Propositions~49--50]{Claudet}. The standard form in those proofs is defined relative to the chosen cover \cite[Definition~19]{Claudet}; it does not require the family of all minimal local sets.

The generalized step can be performed first. An ordinary complementation at $w\in W$ changes edges only among $N(w)$, which is disjoint from $X$, and hence preserves all neighborhoods of $X$. A generalized operation on $X$ changes only $Z$--$Z$ edges and preserves $N(w)$. Each operation therefore preserves both the validity and the edge changes of the other. This remains true throughout a sequence on $W$, since no such operation creates an edge between $W$ and $X$. The reverse implication in Theorem~\ref{thm:standardreduction} holds because all the permitted operations are $\LC_r$ transformations.

To find the allowed generalized operations, we solve the incidence conditions for their vertex multiplicities. This calculation applies to any graph $G$ on at most $n$ vertices and any prescribed independent set $X$. In the reduction, $G=G_1$ and $X$ is its type-$\mathsf X$ set. Put $R=\Z/q\Z$, $q=2^r$, and $J=V\setminus X$, with $m=|X|$ and $t=|J|$. We represent multiplicities modulo $q$ by vectors $s\in\{0,\ldots,q-1\}^X$, equivalently by elements of $R^X$.

For each $v\in J$, let $a_v\in\{0,1\}^X$ record its neighbors in $X$. Write $\odot$ for coordinatewise multiplication, $(z\odot z')_x=z_xz'_x$. For $K\subseteq J$, define
\begin{equation}\label{eq:monomial}
 a_K=\mathop{\bigodot}_{v\in K}a_v,\qquad
 \eta(k)=\begin{cases}1,&k=2,\\ k-2,&k\geq3.\end{cases}
\end{equation}
The scalar product $a_K\cdot s$ is the common-neighbor multiplicity of $K$. For a pair, Definition~\ref{def:incident} requires divisibility by $2^{r-1}$, equivalently $2a_K\cdot s=0\pmod q$. For $3\leq k=|K|\leq r+1$, the required divisor is $2^{r-k+2}$, giving coefficient $2^{k-2}$. Thus put
\begin{equation}\label{eq:admissible}
 \Sigma=\Sigma(G,X,r)=\{s\in R^X:2^{\eta(|K|)}a_K\cdot s=0\pmod q
       \text{ for }K\subseteq J,\ 2\leq|K|\leq r+1\}.
\end{equation}
Pair and triple rows both carry coefficient $2$; subsequent coefficients are $4,8,\ldots$.

\begin{lemma}\label{lem:support}
The vectors in $\Sigma$, with their standard representatives in $\{0,\ldots,q-1\}^X$, are exactly the $r$-incident multisets supported in $X$ whose multiplicities are below $q$. Every $r$-incident multiset supported in $X$ reduces to such a vector.
\end{lemma}
\begin{proof}
For $K\subseteq J$, the common-neighbor multiplicity is $a_K\cdot s$. When $k=|K|$, condition \eqref{eq:incidence} requires divisibility by $2^{r-\eta(k)}$, exactly equivalent to \eqref{eq:admissible}. This equivalence also holds for $r=1$, when the pair condition is vacuous.

The actual support of a multiplicity vector may be smaller than $X$. If a set $K\subseteq V\setminus\supp(S)$ contains some $x\in X$, no vertex of $X$ is adjacent to every vertex of $K$: independence of $X$ forbids adjacency to $x$. Its common-neighbor multiplicity in $S$ is therefore zero. All such extra conditions in Definition~\ref{def:incident} hold automatically. These observations establish both directions. Reducing multiplicities modulo $q$ preserves all relevant divisibilities.
\end{proof}

The multiplicities also determine the changed edges. Define
\begin{equation}\label{eq:togglemap}
 f:\Sigma\longrightarrow\F_2^{\binom J2},\qquad
 f(s)_p=\frac{2a_p\cdot\widetilde{s}}{q}\pmod2,
 \quad p\in\binom J2,
\end{equation}
where $\widetilde{s}$ is any integer vector reducing to $s$ modulo $q$ (an integer lift), and put $\Omega=f(\Sigma)$. The symbols $\Sigma,\Omega,f$ depend on $(G,X,r)$; these arguments will be suppressed when fixed.

\begin{lemma}\label{lem:homomorphism}
For $s,s'\in\Sigma$ and $c\in R$, the map in \eqref{eq:togglemap} is well-defined and satisfies
\begin{equation}\label{eq:homomorphism}
 f(s+s')=f(s)+f(s'),\qquad f(cs)=(c\bmod2)f(s).
\end{equation}
Consequently $\Omega$ is a binary linear subspace.
\end{lemma}
\begin{proof}
The pair constraint is $2a_p\cdot s=0\pmod q$, so every integer lift has $a_p\cdot\widetilde{s}$ divisible by $q/2$. For $r=1$, this is automatic because $q/2=1$. Changing the lift by $qz$ changes the quotient by $2a_p\cdot z$, an even integer. The map is therefore well-defined.

If $\widetilde{s+s'}=\widetilde s+\widetilde{s'}-qz$, then
\begin{equation}\label{eq:carry}
 \frac{2a_p\cdot\widetilde{s+s'}}q
 =\frac{2a_p\cdot\widetilde s}q+\frac{2a_p\cdot\widetilde{s'}}q-2a_p\cdot z.
\end{equation}
Reducing modulo two removes the last term. For a lift $\widetilde c$ of a scalar $c\in R$, writing $\widetilde{cs}=\widetilde c\,\widetilde s-qz$ gives the scalar identity in the same way. Reduction $R\to\F_2$ is well-defined because $q$ is even, establishing \eqref{eq:homomorphism}.
\end{proof}

For $s\in\Sigma$, $f(s)$ is precisely the edge change produced by its generalized complementation. Pairs meeting $X$ are unchanged: even if a vertex of $X$ is absent from the actual support, independence of $X$ makes its common-neighbor multiplicity with any other vertex zero. Thus reduction of multiplicities modulo $q$ preserves both validity and the resulting graph. In a surviving standard-form pair, $\Omega$ is supported on $\binom Z2$, since $N(X)\subseteq Z$.

In \cite[Lemma~19]{CP}, the computation of $\Sigma$ lists the incidence rows for every subset of size at most $r+1$. The number of these rows gives the quasipolynomial cost when $r$ grows logarithmically with the graph order.

\begin{proposition}[Reduction to the edge-change space]\label{prop:reduction}
Let $G_1,G_2$ be a surviving connected pair in common standard form, and let $r\geq1$. Suppose $\omega_1,\ldots,\omega_d$ is a basis of the complete space $\Omega(G_1,X,r)$, with $f(s_i)=\omega_i$. Let $E_i$ be the symmetric zero-diagonal matrix representing $\omega_i$. Then $G_1,G_2$ are $\LC_r$-equivalent if and only if, for some $\epsilon\in\F_2^d$, the graph with adjacency
\begin{equation}\label{eq:changedgraph}
 \Gamma(G_1)+\Delta(\epsilon),\qquad
 \Delta(\epsilon)=\sum_{i=1}^d\epsilon_i E_i,
\end{equation}
can be transformed into $G_2$ by local complementations centered on $W$.
\end{proposition}
\begin{proof}
Theorem~\ref{thm:standardreduction} gives the forward implication, since the generalized step acts by an element of $\Omega$. Conversely, $s=\sum_i\epsilon_i s_i\bmod q$ realizes the selected edge change by Lemma~\ref{lem:homomorphism}. This step and the subsequent complementations are all $\LC_r$ transformations. The argument includes $d=0$, when $\Delta=0$.
\end{proof}

Section~\ref{sec:direct} solves this restricted graph problem in $O(n_c^4)$ binary operations once the basis is known.

\section{Computing admissible edge changes}\label{sec:compression}

\subsection{Incidence compression}

Computing the edge-change space $\Omega$ begins with a compact description of the valid multiplicities $\Sigma$. Their incidence conditions are linear equations modulo $2^r$, but the direct formulation lists a row for every subset of up to $r+1$ vertices. Yet multiplying a row coordinatewise by $2a_v$ produces the next-order row. We apply these maps to a generating set and use modular elimination after each round to discard redundant rows. Linearity preserves all constraints under this compression, so the calculation stores at most $|X|$ generators instead of the full subset list.

A row module over $R$ consists of all $R$-linear combinations of its generating rows. Index $J$ by $[t]$, keeping the vectors and weights of \eqref{eq:monomial}, and let
\begin{equation}\label{eq:module}
 M_r=\spanR_R\{2^{\eta(|K|)}a_K:2\leq|K|\leq r+1\}\subseteq R^m.
\end{equation}
Subsets larger than $t$ contribute no rows. We generate $M_r$ from pair and triple rows by diagonal maps, compressing between rounds. This avoids constructing the higher-order rows individually. The calculation works for arbitrary binary vectors $a_1,\ldots,a_t\in\{0,1\}^m$, with bit complexity polynomial in $m,t,r$. The graph application uses the logarithmic value of $r$ from Section~\ref{sec:reduction}.

\begin{theorem}[Polynomial incidence compression]\label{thm:compression}
A generating matrix for $M_r$ with at most $m$ rows can be constructed deterministically in time polynomial in $m,t,r$. The construction uses at most
\begin{equation}\label{eq:roundbound}
 \tau(r,t)=\max\{0,\min\{r-2,t-3\}\}
\end{equation}
closure rounds, retaining at most $m$ rows with $O(r)$ bits per entry. For $m>0$, a bound for its bit complexity is
\begin{equation}\label{eq:compressioncost}
 O\!\left(\bigl[t^3m^2+(\tau(r,t)+1)m^3(t+1)\bigr]r^2\right).
\end{equation}
\end{theorem}

For $|K|\geq3$ and $v\notin K$, the next row $2^{|K|-1}a_{K\cup\{v\}}$ is obtained from $2^{|K|-2}a_K$ by doubling and coordinatewise multiplication by $a_v$. This gives the $R$-linear maps
\begin{equation}\label{eq:diagonal}
 T_v(z)=2a_v\odot z.
\end{equation}
Let $\widehat M_r$ be the smallest submodule containing $2a_K$ for $|K|=2$ or $3$ and invariant under all the maps $T_v$.

\begin{lemma}\label{lem:closure}
$M_r=\widehat M_r$.
\end{lemma}
\begin{proof}
The pair seeds occur in \eqref{eq:module}. The triple seeds occur when $r\geq2$; when $r=1$, all seeds are zero modulo two. Thus every seed belongs to $M_r$.

We check invariance on a generating row $2^{\eta(k)}a_K$, where $k=|K|$. If $v\in K$, the binary coordinates satisfy $a_v\odot a_K=a_K$, so $T_v$ gives twice that row. If $v\notin K$ and $k=2$, the result is $4a_{K\cup\{v\}}$, twice the triple seed. For $r=1$ this result is zero. If $v\notin K$ and $k\geq3$, then
\begin{equation}\label{eq:raising}
 T_v(2^{k-2}a_K)=2^{k-1}a_{K\cup\{v\}},
\end{equation}
which is the next prescribed row, unless $k=r+1$, in which case it vanishes modulo $2^r$. Hence $M_r$ is invariant and $\widehat M_r\subseteq M_r$.

Conversely, if $|K|=k\geq4$, choose three vertices of $K$ and start with their triple seed. Applying $T_v$ for the remaining $k-3$ distinct vertices gives coefficient $2\cdot2^{k-3}=2^{k-2}$. Pair and triple generators are seeds, so every generator in \eqref{eq:module} belongs to $\widehat M_r$. Hence $M_r\subseteq\widehat M_r$, proving equality.
\end{proof}

Iteration can be necessary even at level three. Take $r=3$, $t=4$, and index the $m=16$ coordinates by $b\in\{0,1\}^4$, with $a_i(b)=b_i$. Every pair or triple seed is annihilated, modulo eight, by the linear functional
\begin{equation}\label{eq:strictfunctional}
 \lambda(z)=\sum_{b\in\{0,1\}^4}(-1)^{b_1+b_2+b_3+b_4}z(b)\pmod8.
\end{equation}
For a seed indexed by $K$, choose $j\notin K$ and pair each $b$ with $b\oplus e_j$, where $e_j$ is the $j$th unit vector and $\oplus$ is binary addition. The two monomial values agree, while their signs in $\lambda$ are opposite. But
\begin{equation}\label{eq:strictexample}
 T_4(2a_{\{1,2,3\}})=4a_{\{1,2,3,4\}},\qquad
 \lambda(4a_{\{1,2,3,4\}})=4\ne0\pmod8.
\end{equation}
Thus the first closure step strictly enlarges the row module generated by the seeds. No higher rows remain at this level, so one step suffices, attaining $\tau(3,4)=1$.

\subsection{Modular row reduction}

A closure round applied to $m$ generating rows can produce $m(t+1)$ rows. Modular elimination reduces them to at most $m$ generators before the next round, while keeping every entry a residue modulo $2^r$.

\begin{lemma}\label{lem:modular}
Let $C$ be a $u\times m$ matrix over $R=\Z/2^r\Z$, with $m,r\geq1$. Its row module has a generating set of at most $m$ rows, computable in $O((u+m)m^2r^2)$ bit operations. All stored entries have $O(r)$ bits. The computation also returns an invertible matrix $\mathsf V$ and a diagonal rectangular matrix $\mathsf D$ such that
\begin{equation}\label{eq:modulardiagonal}
 \mathsf D=\mathsf U C\mathsf V
\end{equation}
for an invertible row transformation $\mathsf U$. The nonzero diagonal entries of $\mathsf D$ are powers of two.
\end{lemma}
\begin{proof}
In the remaining submatrix, select a nonzero entry of smallest $2$-adic valuation and move it to the next diagonal position. Write its standard representative as $2^e a$ with $a$ odd. Multiplying its row by the inverse of $a$ modulo $2^r$ makes the pivot $\sigma=2^e$. Every remaining entry is divisible by $\sigma$. Subtracting multiples of the pivot row clears its column, and subtracting multiples of the pivot column clears its row. These are invertible elementary operations over $R$. Repeat until the remaining submatrix is zero.

To clear an entry, take the exact quotient of its integer representative by $\sigma$. Only the odd factor $a$ is inverted in $R$. Reduce all entries modulo $q=2^r$ after each operation, and track $\mathsf V$ and $\mathsf V^{-1}$ under the column operations. The row transform $\mathsf U$ need not be stored. If the nonzero pivots are $\sigma_1,\ldots,\sigma_k$, the rows
\begin{equation}\label{eq:rowgenerators}
 \sigma_j e_j^{\mathsf T}\mathsf V^{-1},\qquad 1\leq j\leq k,
\end{equation}
generate the row module of $C$, by \eqref{eq:modulardiagonal}.

There are at most $m$ pivot stages, each using $O(um+m^2)$ modular operations, including the search for a pivot and the updates of both transformation matrices. Classical multiplication, division, and inversion of an odd residue cost $O(r^2)$ bit operations. This proves the bound. Empty inputs return the empty generating set and identity column transform.
\end{proof}

\begin{proof}[Proof of Theorem~\ref{thm:compression}]
For $m=0$, return the empty matrix. Otherwise let $M^{(0)}$ be the row module of the pair and triple seeds and iterate
\begin{equation}\label{eq:moduleiteration}
 M^{(j+1)}=M^{(j)}+\sum_{v=1}^t T_v(M^{(j)}).
\end{equation}
By Lemma~\ref{lem:closure}, all these modules lie in $M_r$. A prescribed row indexed by a set of size $k\geq4$ is present after $k-3$ rounds. The largest relevant size is $\min\{t,r+1\}$, so $M^{(\tau(r,t))}=M_r$. For $r\leq2$ or $t\leq3$, the seeds already suffice.

Apply Lemma~\ref{lem:modular} initially and after each round. Since every $T_v$ is $R$-linear, applying it to generators gives generators for its image on the whole module. We can therefore discard redundant rows before the next round without losing a later constraint. The bound $\tau(r,t)$ tells us when to stop, so no canonical basis or equality test is needed. The initial input has $O(t^3)$ rows and $m$ columns, and each later input has at most $m(t+1)$ rows. Lemma~\ref{lem:modular} gives \eqref{eq:compressioncost}; forming the rows costs no more. All retained coefficients remain residues modulo $q$.
\end{proof}

\begin{center}
\fbox{\begin{minipage}{0.92\linewidth}
\textbf{Algorithm 1: weighted incidence compression}
\begin{enumerate}
\item Form the pair and triple rows $2a_K$ modulo $q$.
\item Compress their row module by Lemma~\ref{lem:modular}, obtaining $B$.
\item Repeat $\tau(r,t)$ times: replace $B$ by row generators of the rows of $B$ and all $T_v(B)$.
\item Return $B$.
\end{enumerate}
For $m=0$, return the empty matrix immediately.
\end{minipage}}
\end{center}

\begin{corollary}\label{cor:annihilator}
For $s\in R^m$, the incidence conditions
\begin{equation}\label{eq:allconstraints}
 2^{\eta(|K|)}a_K\cdot s=0\pmod q
 \quad(2\leq |K|\leq r+1)
\end{equation}
hold if and only if $Bs=0$ in $R$.
\end{corollary}
\begin{proof}
Annihilating generators of a row module is equivalent to annihilating the entire module. Apply Theorem~\ref{thm:compression}.
\end{proof}

\subsection{Kernel and image}\label{sec:image}

The matrix $B$ from Theorem~\ref{thm:compression} now represents all incidence constraints. By Corollary~\ref{cor:annihilator}, its kernel is exactly the set of valid multiplicities:
\begin{equation}\label{eq:sigma}
 \Sigma=\ker_R B.
\end{equation}

Applying $f$ to these multiplicities gives the graph changes in $\Omega$. We keep a multiplicity vector realizing each image generator so that the final graph test can return a local unitary as well as a decision.

\begin{theorem}[Complete image with preimages]\label{thm:image}
In polynomial bit complexity in $n,r$, one can compute a basis $\omega_1,\ldots,\omega_d$ of $\Omega$, with $d\leq|X|$, and vectors $s_i\in\Sigma$ satisfying $f(s_i)=\omega_i$. Membership in $\Omega$ is decidable in polynomial time, with a witnessing multiplicity vector for every member.
\end{theorem}

Claudet and Perdrix compute the modular kernel in \cite[Lemma~19]{CP} using Storjohann's Howell transform \cite[Chapter~4]{Storjohann}. The elementary reduction of Lemma~\ref{lem:modular} supplies the same information here.

\begin{lemma}\label{lem:kernel}
For a $b\times m$ matrix $B$ over $R$, with $b\leq m$, generators of $\ker_R B$ can be computed in $O(m^3r^2)$ bit operations. At most $m$ generators are required.
\end{lemma}
\begin{proof}
Assume $m>0$ and use Lemma~\ref{lem:modular} to obtain $\mathsf D=\mathsf U B\mathsf V$ with nonzero pivots $\sigma_1,\ldots,\sigma_k$. Write $s=\mathsf V y$. The equation $Bs=0$ is equivalent to $\sigma_j y_j=0$ modulo $q$ for $j\leq k$, with no condition on the other coordinates. Hence the kernel is generated by
\begin{equation}\label{eq:kernelgenerators}
 \frac{q}{\sigma_j}\mathsf V e_j\ (1\leq j\leq k),
 \qquad \mathsf V e_j\ (k<j\leq m),
\end{equation}
reduced modulo $q$, with zero vectors omitted. The divisions $q/\sigma_j$ are exact integer divisions. Applying Lemma~\ref{lem:modular} to at most $m$ rows proves the cost bound. For $m=0$, the kernel consists of the empty vector.

For example, $B=(2)$ over $\Z/4\Z$ has kernel $\{0,2\}$. Formula~\eqref{eq:kernelgenerators} handles this nonfree module without choosing a vector-space basis.
\end{proof}

Although $\Sigma$ need not be a free $R$-module, its image is a binary vector space. Lemma~\ref{lem:homomorphism} shows that the images of any full kernel generating set span every possible edge change. Keeping their preimages records how to realize the resulting binary basis.

\begin{proof}[Proof of Theorem~\ref{thm:image}]
Apply Lemma~\ref{lem:kernel} to obtain generators $g_1,\ldots,g_\nu$ of $\Sigma$, with $\nu\leq m$. For each $g_j$ and pair $p$, compute the residue $c_p=a_p\cdot g_j\bmod q$ in $\{0,\ldots,q-1\}$. It belongs to $\{0,q/2\}$ by the pair constraint. Thus $f(g_j)_p=2c_p/q$ is an exact integer division; it does not divide by a nonunit in $R$.

Every $s\in\Sigma$ is an $R$-linear combination $\sum_j c_jg_j$. By Lemma~\ref{lem:homomorphism},
\begin{equation}\label{eq:imagegeneration}
 f(s)=\sum_j(c_j\bmod2)f(g_j).
\end{equation}
Conversely, every binary combination of the images is the image of the corresponding sum of the generators. Hence these at most $m$ images span all of $\Omega$.

Perform binary Gaussian elimination on them, carrying along their preimages over $R$. Whenever image rows are added, add the corresponding preimages modulo $q$. The invariant $f(s_i)=\omega_i$ survives by \eqref{eq:homomorphism}. Remove dependent zero image rows. If $\omega\in\Omega$, solve $\omega=\sum_i\epsilon_i\omega_i$ over $\F_2$ and return $s=\sum_i\epsilon_i s_i\bmod q$.

There are $\binom{|J|}{2}=O(n^2)$ image coordinates and at most $m$ generators. Image computation uses $O(n^2m^2)$ elementary arithmetic operations. Before reducing a dot product modulo $q$, its value is at most $m(q-1)$, requiring $O(r+\log(m+1))$ bits. Binary elimination uses $O(n^2m^2)$ field operations, and at most $O(m^3)$ additions of $r$-bit witness coordinates. The witness entries therefore have bounded precision throughout. Together with the compression and kernel bounds, these estimates give polynomial bit complexity. A basis of the binary orthogonal complement of $\Omega$ gives an equation description if needed.
\end{proof}

In particular, one valid $r$-local complementation supported in $X$ takes $G$ to a prescribed graph $H$ exactly when the graphs agree on pairs meeting $X$ and their edge difference on $J$ belongs to $\Omega$. Theorem~\ref{thm:image} decides this condition and returns realizing multiplicities.

\section{The local-unitary equivalence algorithm}\label{sec:integration}

\subsection{Restricted local complementation}\label{sec:direct}

Once the admissible edge changes are known, reaching the target still requires choosing a change and a sequence of ordinary complementations on $W$. Both choices can be expressed in the Clifford equations. Their apparent nonlinearities disappear when the permitted edge changes and complementation centers are taken into account.

In this subsection, $n=|V|$; when applied to a component, this is its order $n_c$. For adjacency matrices $\Gamma,\Gamma'$ on the labelled set $V$, introduce binary vectors $\alpha,\beta,\gamma,\delta\in\F_2^V$ and write $\mathsf D_\xi=\diag(\xi)$. At vertex $u$, the matrix $\bigl(\begin{smallmatrix}\alpha_u&\beta_u\\\gamma_u&\delta_u\end{smallmatrix}\bigr)$ describes the Clifford action on binary Pauli labels; determinant one expresses its invertibility. The graphs are LC-equivalent exactly when \cite[Proposition~7]{CP}
\begin{align}
 \Gamma\mathsf D_\beta\Gamma'+\Gamma\mathsf D_\alpha
       +\mathsf D_\delta\Gamma'+\mathsf D_\gamma&=0,\label{eq:cliffordlinear}\\
 \alpha_u\delta_u+\beta_u\gamma_u&=1\quad(u\in V).\label{eq:clifforddet}
\end{align}
Equation~\eqref{eq:cliffordlinear} is linear in the Clifford coordinates, while \eqref{eq:clifforddet} imposes a quadratic determinant condition at each vertex. These conditions are linked: every solution of the linear equations has equal determinants at adjacent target vertices. A constant-determinant theorem for connected nonbinary self-equivalence spaces appears in Bahramgiri and Beigi \cite[Theorem~4.2]{BB}. The binary argument below applies to a pair of graphs without assuming their equivalence.

\begin{lemma}\label{lem:determinant}
For every solution of \eqref{eq:cliffordlinear}, the vector
$h_u=\alpha_u\delta_u+\beta_u\gamma_u$ is constant on each connected component of the graph with adjacency $\Gamma'$.
\end{lemma}
\begin{proof}
Put
\[
 \mathsf M=\begin{pmatrix}\mathsf D_\alpha&\mathsf D_\beta\\
                        \mathsf D_\gamma&\mathsf D_\delta\end{pmatrix},
 \qquad Q=\binom I{\Gamma'},\qquad
 \mathcal J=\begin{pmatrix}0&I\\I&0\end{pmatrix}.
\]
Equation~\eqref{eq:cliffordlinear} gives
$\mathsf M Q=\binom{\mathsf T}{\Gamma\mathsf T}$, where
$\mathsf T=\mathsf D_\alpha+\mathsf D_\beta\Gamma'$.
Since $\Gamma$ is symmetric,
$(\mathsf M Q)^{\mathsf T}\mathcal J(\mathsf M Q)=0$.
On the other hand,
\[
 \mathsf M^{\mathsf T}\mathcal J\mathsf M
       =\begin{pmatrix}0&\mathsf D_h\\\mathsf D_h&0\end{pmatrix}.
\]
Consequently $\mathsf D_h\Gamma'+\Gamma'\mathsf D_h=0$. Its $uv$ entry is
$(h_u+h_v)\Gamma'_{uv}$, so adjacent vertices have equal determinants.
\end{proof}

For a connected target, one may therefore solve \eqref{eq:cliffordlinear} with any additional affine linear constraints by fixing the four coordinates at one chosen vertex to each of the six invertible binary $2\times2$ matrices. Every solution in any resulting affine space satisfies \eqref{eq:clifforddet} everywhere. The six linear systems take $O((n^2+\ell)n^2)$ binary operations for $\ell$ extra equations. This argument imposes no even-degree hypothesis, unlike the constrained-LC test in \cite[Proposition~10]{CP}. It gives the feasible set as a union of at most six affine spaces, including for the Class $\alpha$ graphs discussed there \cite[Remark~9]{CP}. It applies to one connected component at a time.

\begin{lemma}\label{lem:restrictedsequence}
For arbitrary graphs $G,H$ on $V$ and $W\subseteq V$, a solution of \eqref{eq:cliffordlinear}--\eqref{eq:clifforddet} with $\supp(\beta)\subseteq W$ exists if and only if $H$ can be obtained from $G$ by complementations centered on $W$. A given solution yields such a sequence of length at most $\lfloor3|W|/2\rfloor$ in $O(n^3)$ binary operations.
\end{lemma}

The coordinate updates proving this lemma are given in Appendix~\ref{app:clifford}; they require no connectivity assumption.

For a symmetric zero-diagonal binary matrix $\Delta$ on $V$, write $G\oplus\Delta$ for the graph with adjacency $\Gamma(G)+\Delta$, where addition is over $\F_2$. Now allow edge changes within $F=V\setminus W$. Restricting complementation centers to $W$ fixes the relevant Clifford coordinates on $F$. This makes the unknown edge change enter the graph equations linearly. If $F$ is nonempty and the target is connected, it also fixes determinant one at a vertex and hence everywhere.

\begin{theorem}[Affine graph reduction]\label{thm:direct}
Let $G,H$ be graphs on $V$, with $H$ connected, and fix a partition $V=F\sqcup W$ with $F\ne\varnothing$. Let $E_1,\ldots,E_d$ be symmetric zero-diagonal binary matrices supported on $F\times F$. There is an $\epsilon\in\F_2^d$ for which $G\oplus\Delta(\epsilon)$ reaches $H$ by complementations on $W$ if and only if the affine linear system
\begin{align}
 \Gamma\mathsf D_\beta\Gamma'+\Gamma\mathsf D_\alpha
       +\mathsf D_\delta\Gamma'+\mathsf D_\gamma
       +\sum_i\epsilon_i E_i&=0,\label{eq:directlinear}\\
 \beta|_F=0,\qquad \alpha|_F=\delta|_F&=1\label{eq:frozen}
\end{align}
is consistent, where $\Gamma=\Gamma(G)$ and $\Gamma'=\Gamma(H)$. After substituting \eqref{eq:frozen}, the system has $p=3|W|+n+d$ variables and $n^2$ equations. It can be solved, with a realizing complementation sequence, in $O(n^2p^2+n^3)$ binary operations.
\end{theorem}
\begin{proof}
The Clifford equation for $G\oplus\Delta(\epsilon)$ and $H$ is
\[
 (\Gamma+\Delta)\mathsf D_\beta\Gamma'
 +(\Gamma+\Delta)\mathsf D_\alpha
 +\mathsf D_\delta\Gamma'+\mathsf D_\gamma=0.
\]
Under \eqref{eq:frozen}, $\Delta\mathsf D_\beta=0$ and $\Delta\mathsf D_\alpha=\Delta$, so this is exactly \eqref{eq:directlinear}. Any solution has determinant one on $F$. Lemma~\ref{lem:determinant}, applied to its source graph $G\oplus\Delta$ and the fixed connected target $H$, gives determinant one everywhere. Lemma~\ref{lem:restrictedsequence} supplies the required sequence on $W$.

Conversely, a sequence on $W$ gives a Clifford solution with $\beta|_F=0$. Equation~\eqref{eq:clifforddet} then forces $\alpha|_F=\delta|_F=1$. Its linear equation expands to \eqref{eq:directlinear}. Forming that system costs $O(n^3+n^2d)$ operations; dense elimination costs $O(n^2p^2)$, and sequence extraction costs $O(n^3)$.
\end{proof}

\begin{corollary}\label{cor:feasible}
Under the hypotheses of Theorem~\ref{thm:direct}, the set of feasible coefficients $\epsilon$ is empty or affine. An origin and a basis can be computed within the same time bound. If the $E_i$ are linearly independent and this affine space has dimension $\rho$, exactly $2^\rho$ distinct edge changes admit a complementation sequence on $W$ leading to $H$.
\end{corollary}
\begin{proof}
Project the affine solution space of \eqref{eq:directlinear}--\eqref{eq:frozen} onto its $\epsilon$ coordinates. Projecting an origin and a basis, then eliminating dependent vectors, gives the claimed description. Independence of the $E_i$ makes distinct coefficients represent distinct edge changes.
\end{proof}

For the standard-form pair of Section~\ref{sec:reduction}, take $F=X\sqcup Z$ and use the basis supplied by Theorem~\ref{thm:image}. The matrices $E_i$ are supported on $Z\times Z$. All rows and columns involving $X$ in \eqref{eq:directlinear} vanish after setting $\gamma|_X=0$: the $X$-neighborhoods agree, lie in $Z$, and meet no vertex of $W$. Thus only the equations on $J=Z\sqcup W$ need be retained. They have
\begin{equation}\label{eq:smallvariablecount}
 p=3|W|+|J|+d=4|W|+|Z|+d\leq4n_c
\end{equation}
variables, because $d\leq|X|$. The target graph on the original component, before deleting $X$, is the connected graph used in the determinant proof. The calculation costs $O(n_c^4)$ binary operations. In particular, it avoids the auxiliary-graph construction of \cite[Section~3.3]{CP}, which adds a vertex for each edge of each image generator and can require $O(n_c^3)$ auxiliary vertices. It also works when $d=0$. If $F=\varnothing$, then $X=Z=\varnothing$ and $\Omega=0$; use the six linear systems above for ordinary LC equivalence.

\subsection{Decision and local unitaries}

\begin{center}
\fbox{\begin{minipage}{0.92\linewidth}
\textbf{Algorithm 2: labelled graph-state LU equivalence}
\begin{enumerate}
\item Compare labelled connected-component partitions. Reject if they differ.
\item On each component of order $n_c\geq2$, choose $r$ by \eqref{eq:level} and perform steps 3--5.
\item Compute common standard forms and check their types and $X$-neighborhoods. Reject if these tests fail.
\item Compute $M_r$, $\Sigma$, and a basis of $\Omega$ with multiplicity preimages, using Algorithm~1 and Theorem~\ref{thm:image}.
\item Solve \eqref{eq:directlinear}--\eqref{eq:frozen} with $F=X\sqcup Z$. If $F$ is empty, use the six ordinary-LC systems instead. Reject an inconsistent instance; otherwise retain $\epsilon$ and extract the complementation sequence.
\item Accept when all components pass. Combine the component witnesses with the standardization paths to obtain the local unitaries.
\end{enumerate}
One-vertex components require no transformation.
\end{minipage}}
\end{center}

\begin{proof}[Proof of Theorem~\ref{thm:main}]
By Lemma~\ref{lem:components}, equivalence can be tested separately on matching labelled components. The standard-form procedure and the type and neighborhood comparisons reject only inequivalent inputs. On each surviving component, the chosen level makes LU equivalence coincide with $\LC_r$ equivalence (Theorem~\ref{thm:level}). The complete image from Theorem~\ref{thm:image} includes every generalized step allowed by Proposition~\ref{prop:reduction}, and Theorem~\ref{thm:direct} tests which of them can be followed by the required ordinary complementations. Thus the algorithm accepts exactly the equivalent inputs. This includes components with no nonzero edge changes; when $F$ is empty, the remaining test is ordinary LC equivalence.

For the bit complexity, let $m,t\leq n_c$ and $r=O(\log(n_c+1))$. Equation~\eqref{eq:compressioncost} and $\tau(r,t)\leq r$ give $O(n_c^5r^3)$ bit operations for compression. The modular kernel costs $O(n_c^3r^2)$. Computing the binary image and carrying its preimages costs at most $O(n_c^4(r+\log(n_c+1))^2)$ bit operations. The direct graph test costs $O(n_c^4)$ binary operations and witness extraction costs $O(n_c^3)$. Hence all algebra after standardization costs $\widetilde O(n_c^5)$ bit operations.

Writing $T_{\mathrm{std}}(n_c)$ for the cost of the inherited common-standard-form procedure, the component cost is $T_{\mathrm{std}}(n_c)+\widetilde O(n_c^5)$. Lemma~\ref{lem:standardcost} accounts for the standardization loops and gives
\[
 T_{\mathrm{std}}(n_c)=\widetilde O\!\left(n_c^4 T_{\mathrm{rank}}(n_c)+n_c^6\right),
\]
where $T_{\mathrm{rank}}(n_c)$ is the bit cost of binary matrix rank. An admissible rank exponent below $2.38$ yields $\widetilde O(n_c^{6.38})$, recovering the bound of \cite[Lemma~16]{CP}. Classical cubic rank computation instead gives $\widetilde O(n_c^7)$. The $\widetilde O(n_c^5)$ cost of the new algebra is independent of this choice.

The witness construction below adds no larger cost. Finally, summing over component orders $n_1,\ldots,n_k$ gives $\sum_i n_i^{6.38}\leq n^{6.38}$, with logarithmic factors bounded by those for $n$. The resulting algorithm has the claimed complexity.
\end{proof}

Retain the ordinary-LC paths used to put $G,H$ into standard forms $G_1,G_2$. From a solution $\epsilon$, set
\begin{equation}\label{eq:recovermultiplicity}
 s=\sum_i\epsilon_i s_i\pmod q.
\end{equation}
Proposition~\ref{prop:reduction} identifies its generalized complementation with the selected edge change, and Lemma~\ref{lem:restrictedsequence} constructs the subsequent path from $G_1\star_r s$ to $G_2$. Following the inverse of the second standardization path completes the transformation to $H$.

Equation~\eqref{eq:generalizedunitary} implements the recovered generalized complementation by exact single-qubit rotations. All other graph operations are local Cliffords. Absorbing the Hadamards that conjugate an $X$-rotation to a $Z$-rotation therefore gives
\begin{equation}\label{eq:outputunitary}
 U_v=C_v^{\mathrm{out}}Z(k_v\pi/q)C_v^{\mathrm{in}},
 \qquad 0\leq k_v<2q,
\end{equation}
up to an irrelevant global phase, with one-qubit Cliffords on either side. Products of single-qubit Cliffords can be accumulated using their finite multiplication table modulo phase. Each $k_v$ requires $O(r)$ bits, and the Clifford descriptions have constant size. Thus the total symbolic output has $O(n\log(n+1))$ bits. For components using different values of $q$, take their largest value and rescale the numerators. The underlying Clifford--rotation--Clifford form follows from the prior graphical reduction; the present algorithm computes it in polynomial time.

\begin{proposition}[Checking a transformation]\label{prop:certificate}
Every accepting input has a graph-operation certificate of $O(n^2\log(n+1))$ bits that can be verified in $\widetilde O(n^5)$ bit operations. Verification requires neither an MLS cover nor the standardization procedure, and determines the exact local unitaries realizing the transformation.
\end{proposition}
\begin{proof}
For each labelled component, record the ordinary-LC path from $G$ to $G_1$, the level $r$ and multiplicities $s$ of the generalized step, and the ordinary-LC path from $G_1\star_r s$ to $H$. Lemma~\ref{lem:standardcost} bounds each standardization path by $O(n_c^2)$ local complementations; the remaining restricted path has length $O(n_c)$ by Lemma~\ref{lem:restrictedsequence}. Together with $r=O(\log(n_c+1))$ and $0\leq s(u)<2^r$, this gives the stated certificate size.

The verifier checks the component labels and the bounds on the path lengths, level, and multiplicities. It simulates the first path, checks that $\supp(s)$ is independent, and, when $s\ne0$, applies Theorem~\ref{thm:compression} with $X=\supp(s)$. Testing the resulting generating rows against $s$ modulo $2^r$ verifies all incidence conditions in $\widetilde O(n_c^5)$ bit operations; for $s=0$ they hold automatically. The pairwise common-neighbor sums then determine $G_1\star_r s$. Finally, the verifier simulates the second path and compares the resulting labelled graph with $H$. The ordinary paths cost $O(n_c^4)$ binary operations. Equation~\eqref{eq:generalizedunitary} and the local-Clifford implementation of each ordinary complementation specify the exact unitaries, which can be accumulated in the form \eqref{eq:outputunitary}. Summing over components gives the claimed bound.
\end{proof}

Theorem~\ref{thm:main} places graph-state LU equivalence in $\mathsf P$, and therefore in $\mathsf{NP}$, answering the membership question in \cite[Chapter~8]{Claudet}. Proposition~\ref{prop:certificate} also supplies a direct verifier that does not repeat the graph preprocessing.

The same computation determines how far beyond Clifford operations a transformation must go in the dyadic hierarchy. The standard forms can be reused while the level varies.

\begin{corollary}\label{cor:minlevel}
For two labelled graph states, the least integer $r\geq1$ for which they are $\LC_r$-equivalent, or the conclusion that they are not LU-equivalent, can be computed in $\widetilde O(n^{6.38})$ bit operations.
\end{corollary}
\begin{proof}
The standardization and type comparisons are independent of $r$. After performing them once, apply the algebraic part of Algorithm~2 at
\[
 r_{\max}=\max\{1,\lceil\log_2((n+1)/8)\rceil\}.
\]
If it rejects, Theorem~\ref{thm:level} rules out LU equivalence. Otherwise $\LC_r\subseteq\LC_{r+1}$ makes the fixed-level acceptance predicate monotone. Binary search on $1,\ldots,r_{\max}$ takes $O(\log\log(n+2))$ calls. Each call costs $\widetilde O(n^5)$ after the common preprocessing, so the overall bound is unchanged. For disconnected inputs, the least level is the maximum of the least levels of the components.
\end{proof}

Level one permits only Clifford gates; level two also permits $Z(\pi/4)$. The minimum level concerns exact membership in these gate sets. It does not minimize the number of non-Clifford gates or the length of a circuit.

\subsection{Stabilizer states and codes}\label{sec:codes}

A pure qubit stabilizer state, supplied by signed Pauli generators, can be converted in polynomial time to a local-Clifford-equivalent graph state \cite{VDDGraph,VDD}. Binary elimination performs this conversion in $O(n^3)$ bit operations. Retaining the local Cliffords extends both the decision and synthesis parts of Theorem~\ref{thm:main} to this input representation, with the same asymptotic bound.

Two pure states are equivalent under stochastic local operations and classical communication (SLOCC) if one is a nonzero scalar multiple of $(\bigotimes_v A_v)|\psi\rangle$ for invertible single-qubit operators $A_v$.

\begin{corollary}\label{cor:slocc}
SLOCC equivalence of two pure $n$-qubit stabilizer states given by signed Pauli generators can be decided in $\widetilde O(n^{6.38})$ bit operations. When they are equivalent, the algorithm returns a product of single-qubit unitaries relating them.
\end{corollary}
\begin{proof}
For pure qubit stabilizer states, SLOCC equivalence coincides with LU equivalence \cite[Theorem~2]{VDDLocal}. Apply the stabilizer-to-graph conversion and Theorem~\ref{thm:main}.
\end{proof}

The result also applies to code spaces of dimension two. Let $\mathcal C\subseteq(\C^2)^{\otimes n}$ be the simultaneous $+1$ eigenspace of $n-1$ independent commuting Hermitian Pauli operators whose generated group excludes $-I$. Such a stabilizer code encodes one logical qubit; write $P_{\mathcal C}$ for its orthogonal projector. Equivalence of two such codes means that a product of single-qubit unitaries maps one entire code space to the other. The physical-qubit correspondence is fixed, and the induced logical unitary is unrestricted.

\begin{corollary}\label{cor:onequbitcode}
Given signed Pauli generators for two $n$-qubit stabilizer codes $\mathcal C,\mathcal D$, each encoding one logical qubit, one can decide in $\widetilde O(n^{6.38})$ bit operations whether
\begin{equation}\label{eq:codeequivalence}
 UP_{\mathcal C}U^\dagger=P_{\mathcal D},
 \qquad U=\bigotimes_{v=1}^n U_v,\quad U_v\in U(2).
\end{equation}
If so, the algorithm returns such a product $U$ with an exact symbolic description of $O(n\log(n+1))$ bits.
\end{corollary}
\begin{proof}
Purifying the maximally mixed code state on one reference qubit converts this problem to pure-state LU equivalence. We use the stabilizer purification of Gross and Van den Nest \cite[Section~5]{GrossVDN}. Binary symplectic elimination finds logical Pauli operators $\overline X,\overline Z$ that commute with the code stabilizer and anticommute with each other. Choose a normalized code vector $|\overline0\rangle$ with $\overline Z$ eigenvalue $+1$ and put $|\overline1\rangle=\overline X|\overline0\rangle$. Introduce a reference qubit $R$ and define
\begin{equation}\label{eq:codepurification}
 |\Psi_{\mathcal C}\rangle=
 \frac{|0\rangle_R|\overline0\rangle+
       |1\rangle_R|\overline1\rangle}{\sqrt2}.
\end{equation}
If $S_1,\ldots,S_{n-1}$ generate the code stabilizer, the $n+1$ independent commuting generators
\[
 I_R\otimes S_1,\ldots,I_R\otimes S_{n-1},
 \quad \mathsf X_R\otimes\overline X,
 \quad \mathsf Z_R\otimes\overline Z
\]
specify $|\Psi_{\mathcal C}\rangle$. Thus its tableau is computable in $O(n^3)$ bit operations, without constructing a state vector, and
\[
 \operatorname{Tr}_R|\Psi_{\mathcal C}\rangle\langle\Psi_{\mathcal C}|
       =\tfrac12P_{\mathcal C}.
\]
Construct $|\Psi_{\mathcal D}\rangle$ in the same way, with any choice of logical Pauli pair.

Suppose \eqref{eq:codeequivalence} holds. The states $(I_R\otimes U)|\Psi_{\mathcal C}\rangle$ and $|\Psi_{\mathcal D}\rangle$ purify the same density operator $P_{\mathcal D}/2$. Their Schmidt decompositions therefore differ by a unitary $V_R$ on the reference qubit, giving
\begin{equation}\label{eq:purifiedequivalence}
 (V_R\otimes U)|\Psi_{\mathcal C}\rangle
       =|\Psi_{\mathcal D}\rangle
\end{equation}
up to global phase. Conversely, taking the partial trace of \eqref{eq:purifiedequivalence} yields \eqref{eq:codeequivalence}. Since $R$ is a single qubit, \eqref{eq:purifiedequivalence} is precisely an instance of pure-stabilizer LU equivalence on $n+1$ labelled qubits. Applying Theorem~\ref{thm:main} after graph conversion decides it; retaining the physical factors of its witness gives $U$. The conversion costs $O(n^3)$ and the remaining computation costs $\widetilde O((n+1)^{6.38})$, which gives the stated bound and output size.
\end{proof}

For a stabilizer code encoding $k>1$ logical qubits, this purification uses $k$ reference qubits. The unitary relating two purifications may act jointly on them, whereas Theorem~\ref{thm:main} allows only products of single-qubit unitaries. Thus the same argument does not decide LU equivalence of general stabilizer codes.

\section{LC classes within an LU class}\label{sec:orbitcount}

A single LU class may contain several LC classes. For a fixed labelled graph $G$, we determine how many there are and whether every graph state LU-equivalent to $|G\rangle$ is already LC-equivalent to it. Claudet and Perdrix characterize the latter property through generalized complementations in standard form \cite[Section~4.4]{Graphical}. The compressed edge space and the affine graph equations make this test polynomial and give the exact count when there is more than one class.

\begin{theorem}\label{thm:orbitcount}
The number of LC-equivalence classes of labelled graph states in the LU-equivalence class of an $n$-vertex graph $G$ is $2^\kappa$ for a nonnegative integer $\kappa$. Both $\kappa$ and this count can be computed in $\widetilde O(n^{6.38})$ bit operations. Within the same bound, one can decide whether LU equals LC on this class and, if not, construct a labelled graph $H$ that is LU-equivalent but not LC-equivalent to $G$, together with a local-unitary witness.
\end{theorem}

The count concerns LC classes of graph states with fixed vertex labels. Individual graphs and arbitrary states in the continuous LU orbit are different counting objects.

The pairwise algorithm can choose a standard form using both inputs. Counting a whole class requires a description that works for every target at once. Work first with a connected graph on $n_c\geq3$ vertices. Apply the common-standard-form procedure to $(G,G)$ and retain one output $G_0$, its MLS cover $\mathcal M$, and its partition $X,Z,W$. These choices depend only on the starting graph. Choose $r$ by \eqref{eq:level}. Writing $\Omega_j(G_0,X)=\Omega(G_0,X,j)$ for the edge-change space at level $j$, set
\begin{equation}\label{eq:orbitspaces}
 \Omega=\Omega_r(G_0,X),\qquad
 \Omega_1=\Omega_1(G_0,X)\subseteq\Omega,\qquad
 G_\omega=G_0\oplus\omega.
\end{equation}
Here an edge-change vector is identified with its symmetric zero-diagonal matrix. The inclusion follows by multiplying a level-one multiplicity vector by $2^{r-1}$. The space $\Omega_1$ is generated by the neighborhood cliques of the vertices in $X$.

\begin{lemma}\label{lem:fixedcover}
Every $G_\omega$, $\omega\in\Omega$, is in standard form for the fixed cover $\mathcal M$. Every labelled graph $H$ LU-equivalent to $G_0$ is LC-equivalent to at least one $G_\omega$.
\end{lemma}
\begin{proof}
Let $\mathcal S_G(L)$ be the subgroup of graph-state stabilizers supported in $L$. Its order is determined by the spectrum of the reduced state on $L$, so it is LU-invariant. For a minimal local set $L$, every nonidentity member has support exactly $L$. Restriction to a vertex of $L$ is therefore injective on the binary Pauli labels: a nonidentity element in the kernel would have strictly smaller support. The subgroup has dimension one or two. In the dimension-two case, its three nonidentity members use three different Pauli axes at every vertex, and their commutation forces $|L|$ to be even. These are the minimal-support constraints underlying the standard-form reduction \cite[Section~4.1]{Graphical}.

Suppose $U=\bigotimes_v U_v$ maps $|G_0\rangle$ to $|H\rangle$. For a dimension-one member $L\in\mathcal M$, the reduced state is $2^{-|L|}(I+P_L)$. Thus $U_v$ sends the Pauli axis of $P_L$ to a signed Pauli axis for every $v\in L$. Two distinct such axes at one vertex force $U_v$ to be Clifford.

For a dimension-two member, $|L|=2$ would give a pure two-qubit reduced state, contradicting connectedness and $n_c\geq3$. Hence $|L|\geq4$. Write its three nonidentity stabilizers as $P_1,P_2,P_3$. At every vertex their local factors are three distinct Pauli axes. For any two vertices $u,v\in L$, the restrictions of $P_1,P_2,P_3$ to $L\setminus\{u,v\}$ are nonzero, mutually orthogonal Pauli tensors. Consequently, the operator Schmidt support of the nonidentity part of the reduced state on the $\{u,v\}$ side is exactly the span of the three products $P_j|_u\otimes P_j|_v$.

The target reduced state has the same property. In source and target Pauli coordinates chosen according to these triples, the two supports are the space of diagonal $3\times3$ matrices. Conjugation by $U_u,U_v$ acts through real orthogonal matrices $R_u,R_v$, so it maps each source product to
\[
 (R_u e_j)(R_v e_j)^{\mathsf T},\qquad j=1,2,3,
\]
where $e_j$ is a coordinate vector. Each matrix lies in the target diagonal space and has rank one. A nonzero diagonal rank-one matrix has exactly one nonzero diagonal entry; thus both $R_u e_j$ and $R_v e_j$ lie along the corresponding coordinate axis. Orthogonality makes these axes distinct as $j$ varies. Hence $R_u,R_v$ are signed permutation matrices, which means that $U_u,U_v$ normalize the Pauli operators and are Clifford up to phase. Varying the pair proves this at every vertex of $L$.

Every vertex of $W$ is constrained by more than one Pauli axis, so $U_v$ is Clifford there. On $X$ and $Z$, factor out a Clifford that sends the source Pauli axis to its signed image. The remaining operator commutes with the source axis. Up to global phase,
\begin{equation}\label{eq:residualrotation}
 |H\rangle=C U'|G_0\rangle,\qquad
 U'=\prod_{x\in X}X_x(\theta_x)\prod_{z\in Z}Z_z(\theta_z),
\end{equation}
where $C$ is local Clifford, $\theta_v\in\mathbb R$, and $U'$ acts trivially on $W$. Only constraints from the chosen cover were used.

Since $X$ is independent and $N_{G_0}(X)\subseteq Z$, the stabilizer identities $\mathsf X_x|G_0\rangle=\mathsf Z_{N_{G_0}(x)}|G_0\rangle$ can be applied simultaneously to the rotations in $U'$. For $y\in\F_2^{V}$ this gives
\begin{equation}\label{eq:residualphase}
 \frac{\langle y|U'|G_0\rangle}{\langle y|G_0\rangle}
 =\exp\!\left(i\sum_{x\in X}\theta_x
                  \left(\bigoplus_{z\in N_{G_0}(x)}y_z\right)
                +i\sum_{z\in Z}\theta_z y_z\right).
\end{equation}
Here the parity is interpreted as $0$ or $1$. The ratio depends only on $y_Z$ and has unit modulus. Moreover, $U'|G_0\rangle=C^\dagger|H\rangle$ is a stabilizer state with full computational support.

Any full-support stabilizer state has amplitudes, up to a common scalar, of the form
\begin{equation}\label{eq:fullsupportphase}
 i^{\sum_v \ell_v y_v}(-1)^{\sum_{u<v}b_{uv}y_uy_v},
 \qquad \ell_v\in\Z/4\Z,\quad b_{uv}\in\F_2.
\end{equation}
This is the full-support case of the stabilizer normal form \cite[Theorem~5]{GrossVDN}. Directly, full support excludes a nonidentity diagonal stabilizer, so projection of the binary stabilizer onto its $X$ coordinates is an isomorphism. Choose a generator with $X$ coordinate $e_v$ for each $v$. Commutation makes the off-diagonal $Z$ coefficients symmetric; the amplitude recurrences give \eqref{eq:fullsupportphase}.

Apply this normal form to $U'|G_0\rangle$ and divide by the graph phase of $G_0$. The ratio in \eqref{eq:residualphase} depends only on $y_Z$. Evaluating it at $e_v$ shows that its linear coefficient vanishes modulo four outside $Z$; evaluating at $e_u+e_v$ shows that its quadratic coefficient vanishes whenever an endpoint is outside $Z$. Therefore
\begin{equation}\label{eq:fixedgraphconversion}
 U'|G_0\rangle=D_Z|K\rangle,\qquad K=G_0\oplus\Delta_Z,
\end{equation}
where $D_Z$ is a product of diagonal single-qubit Cliffords on $Z$, and $\Delta_Z$ changes only $Z$--$Z$ edges. In particular, $K$ is LC-equivalent to $H$.

The chosen cover is preserved by this conversion. Indeed, $D_Z^\dagger U'$ commutes with every element of $\mathcal S_{G_0}(L)$ for every $L\in\mathcal M$: its factors on $X$ and $Z$ have the fixed Pauli axes, and its factors on $W$ are identities. These supported stabilizer groups are therefore contained in those of $K$. LU invariance gives equality of their orders, hence equality of the groups. Minimality is also LU-invariant, so $\mathcal M$ is an MLS cover of $K$ with the same types. The $X$ neighborhoods and vertex order have not changed, and $K$ is in standard form for the same cover. Applying this argument to \eqref{eq:generalizedunitary} also proves that every $G_\omega$ has that standard form.

Finally, Theorem~\ref{thm:level} gives $G_0\sim_{\LC_r}K$. The fixed-level standard-form reduction \cite[Lemma~18]{CP} expresses $K$ as $G_\omega$ followed by complementations on $W$, for some $\omega\in\Omega$. Therefore $H\sim_{\LC}G_\omega$. This step quantizes the graph transformation after \eqref{eq:fixedgraphconversion}; no dyadic assumption was made on the angles in \eqref{eq:residualrotation}.
\end{proof}

\begin{proof}[Proof of Theorem~\ref{thm:orbitcount}]
For the fixed connected standard form, define
\begin{equation}\label{eq:lccancellablespace}
 \Omega_{\LC}=\{\omega\in\Omega:G_\omega\sim_{\LC}G_0\}.
\end{equation}
All $G_\omega$ have the same cover and the same $X$ neighborhoods. The fixed-level reduction at level one therefore says that
\[
 G_\omega\sim_{\LC}G_{\omega'}
 \quad\Longleftrightarrow\quad
 T_W(G_0\oplus\omega\oplus\lambda)=G_0\oplus\omega'
 \text{ for some }\lambda\in\Omega_1,
\]
where $T_W$ is a sequence of complementations on $W$. A fixed $Z$--$Z$ edge change commutes with this sequence: before each step the two graphs have identical neighborhoods at its center, and afterwards they still differ by exactly that edge change. Translating both sides by $\omega'$ gives
\begin{equation}\label{eq:orbitcosets}
 G_\omega\sim_{\LC}G_{\omega'}
 \quad\Longleftrightarrow\quad
 \omega+\omega'\in\Omega_{\LC}.
\end{equation}

To compute $\Omega_{\LC}$, allow both the proposed change in $\Omega$ and any ordinary change in $\Omega_1$ that may help undo it. Choose bases $E_1,\ldots,E_d$ of $\Omega$ and $E^{(1)}_1,\ldots,E^{(1)}_{d_1}$ of $\Omega_1$, with binary coefficient vectors $\epsilon\in\F_2^d$ and $\epsilon'\in\F_2^{d_1}$. Apply Theorem~\ref{thm:direct} with source and target both $G_0$, frozen set $F=X\sqcup Z$, and
\begin{equation}\label{eq:orbitprojection}
 \Delta(\epsilon,\epsilon')=\sum_{i=1}^{d}\epsilon_i E_i
                      +\sum_{j=1}^{d_1}\epsilon'_j E^{(1)}_j.
\end{equation}
The $\epsilon'$ variables account for ordinary complementations at $X$, while the graph system accounts for those at $W$. Projecting its affine solution space onto $\epsilon$ gives exactly $\Omega_{\LC}$ under the basis identification with $\Omega$. The identity transformation puts zero in this projection, making it a vector space. The combined list in \eqref{eq:orbitprojection} need not be independent; only the independent $E_i$ coordinates are counted after projection. If $F$ is empty, $\Omega=0$ and Lemma~\ref{lem:fixedcover} already gives one LC class.

By Lemma~\ref{lem:fixedcover} and \eqref{eq:orbitcosets}, the LC classes of graph states in this LU class are in bijection with
\begin{equation}\label{eq:orbitquotient}
 \Omega/\Omega_{\LC},\qquad
 \kappa=\dim\Omega-\dim\Omega_{\LC}.
\end{equation}
Thus the count is $2^\kappa$. If $\kappa>0$, take any vector outside $\Omega_{\LC}$ and lift it to multiplicities by Theorem~\ref{thm:image}. The corresponding $G_\omega$ is the required counterexample, with its unitary given by \eqref{eq:generalizedunitary} and the source standardization path.

The two edge spaces cost $\widetilde O(n_c^5)$ after standardization. Since $d,d_1\leq|X|$, the projected graph system has $O(n_c)$ variables and costs $O(n_c^4)$ binary operations. The total remains $\widetilde O(n_c^{6.38})$. Connected components of one or two vertices each contribute one LC class. For a disconnected graph, Lemma~\ref{lem:components} makes the labelled component partitions invariant, and LC transformations also act componentwise. The counts multiply and the exponents add. The sum of the component costs has the asserted bound, while $\kappa\leq n$ ensures that writing the integer count takes only $O(n)$ bits.
\end{proof}

Ordinary complementations at $W$ can undo edge changes that ordinary complementations at $X$ cannot generate. Thus quotienting only by $\Omega_1$ can overcount the LC classes. For an explicit example, start with the 27-vertex graph described in Appendix~\ref{app:checks}, choose two vertices $z_1,z_2\in Z$, and add four vertices $u_1,u_2,u_3,u_4$ with edges
\[
 u_1u_4,\ u_2u_4,\ u_3u_4,\quad
 z_1u_1,\ z_1u_2,\ z_2u_1,\ z_2u_3.
\]
The old $X$ stars, together with $\{u_1,u_2,u_3,u_4\}$, form an MLS cover. The added set has dimension two: its supported stabilizers are generated by $\mathsf X_{u_1}\mathsf X_{u_2}\mathsf X_{u_3}\mathsf Z_{u_4}$ and $\mathsf Z_{u_1}\mathsf Z_{u_2}\mathsf Z_{u_3}\mathsf X_{u_4}$. Every nonidentity member has full support, so all four vertices have type $\bot$.

Complementing $u_1,u_2,u_3$ in that order toggles only $z_1z_2$. For this 31-vertex graph, $\dim\Omega=15$, $\dim\Omega_1=14$, and $\dim\Omega_{\LC}=15$. Thus its LU class contains one LC class, although counting modulo $\Omega_1$ alone would give two. The exact checks and realizing sequences are recorded in Appendix~\ref{app:checks}.

Counting Clifford classes and counting graph representatives have different complexities. The former is polynomial by Theorem~\ref{thm:orbitcount}; the decision algorithm completes the complexity classification of the latter.

\begin{corollary}\label{cor:graphcount}
Given a graph $G$ on $[n]$, counting the graphs $H$ on the same labelled vertex set whose graph states are LU-equivalent to $|G\rangle$ is $\#\mathsf P$-complete. The problem remains $\#\mathsf P$-hard when the input is restricted to circle graphs.
\end{corollary}
\begin{proof}
Dahlberg, Helsen, and Wehner proved $\#\mathsf P$-completeness for counting LC-equivalent graph states, with hardness already on circle graph inputs \cite{DHWCount}. Hahn et al.\ use LU=LC for circle graph states to obtain LU-counting hardness \cite[Corollary~5]{Circle}. For membership in $\#\mathsf P$, guess the $\binom n2$ adjacency bits of $H$ and run the deterministic test of Theorem~\ref{thm:main}. Each accepted labelled graph has one encoding and contributes one accepting branch.
\end{proof}

\section{Discussion}\label{sec:discussion}

Suppose a device can prepare $|G\rangle$ more easily than the desired graph-state resource $|H\rangle$. We can test whether single-qubit operations transform one into the other and, if so, compute those operations exactly. In MBQC, the corrections can be folded into the measurement bases whenever the required bases are available. The returned rotations also make the control requirements explicit for a given device.

Graph-state compilers already use locally equivalent representations to reduce preparation costs \cite{VijayanCompiler}, and photonic fusion schemes search within LC classes for less costly resources \cite{JiFusion}. The quotient in Theorem~\ref{thm:orbitcount} could extend such searches across distinct LC classes of the same LU class. Its description has polynomial size even when the number of classes is exponential. Finding a class with lower preparation depth or fusion overhead without enumerating those classes remains an optimization problem. The local corrections are part of its cost: savings in entangling operations have to be weighed against their implementation and calibration, especially for non-Clifford rotations in a fault-tolerant architecture. Minimizing the number of non-Clifford operations is also a separate objective from finding the least dyadic level.

The reference implementation starts from the two adjacency matrices, computes an MLS cover and common standard forms, and returns a transformation certificate. Its tests include independently scrambled versions of a known LU-but-not-LC pair and a separate certificate checker (Appendix~\ref{app:checks}). Benchmarks on larger structured resource families should time preprocessing separately from the modular and binary stages; the current code uses classical rank routines. On components of at most 26 qubits, ordinary LC testing already suffices \cite{Minimal27}, offering a shortcut for smaller instances.

The worst-case exponent is set by standardization. The new algebra costs $\widetilde O(n^5)$ bit operations after this stage, while the inherited common-standard-form procedure determines the overall exponent $6.38$. A faster way to obtain these standard forms would therefore improve the full bound.

For stabilizer codes, Corollary~\ref{cor:onequbitcode} relies on a purification with one reference qubit. With several logical qubits, the unitary relating purifications can act jointly on the reference system, so the present product-unitary test no longer suffices. Extending the algorithm would require handling this joint action.

Experimental states also raise an approximate version of the problem. One would need to specify how the imperfect states are represented and how close a local transformation must bring them. The exact test settles the relation between ideal resources; preparation and control errors still determine whether the computed correction produces the desired experimental fidelity.

\section*{Acknowledgements}

This project has received funding from the European Research Council (ERC) under the European Union's Horizon 2020 research and innovation programme (grant agreement No.~864597). This research was supported in part by grant NSF PHY-2309135 to the Kavli Institute for Theoretical Physics (KITP).

This work is part of the author's experimental series on agentic scientific research \cite{AgenticSeries}. The author set the research directions, designed the harness that coordinated model tasks, tools, memory, and checking, supplied physical and mathematical intuitions, and guided the development and revision of the research ideas. OpenAI models developed detailed candidate arguments and derivations, searched the literature, wrote and ran computational checks, and drafted and revised the manuscript under this supervision.

\appendix
\section{Connected components}\label{app:reduction}

\begin{lemma}\label{lem:components}
Two labelled graph states are LU-equivalent exactly when their connected-component partitions agree and the corresponding component states are LU-equivalent.
\end{lemma}
\begin{proof}
For $A\subseteq V$, let $A^c=V\setminus A$. The stabilizers supported entirely in $A$ are parametrized by $d_A\in\ker_{\F_2}\Gamma_{A^c,A}$: the product of graph-state generators indexed by $d_A$ has no Pauli support on $A^c$ precisely when $\Gamma_{A^c,A}d_A=0$. For the reduced state $\rho_A=\operatorname{Tr}_{A^c}|G\rangle\langle G|$, expansion in these orthogonal Pauli operators gives
\begin{equation}\label{eq:cutpurity}
 \operatorname{Tr}(\rho_A^2)
 =2^{-|A|}\bigl|\ker_{\F_2}\Gamma_{A^c,A}\bigr|
 =2^{-\operatorname{rank}_{\F_2}\Gamma_{A^c,A}}.
\end{equation}
Thus the state factors across $A:A^c$ if and only if $\Gamma_{A^c,A}=0$, equivalently, no edge crosses the cut.

Product unitaries preserve factorization across every fixed labelled cut. The finest partition into graph-state product factors is thus the connected-component partition, which must be the same for both inputs. When partitions coincide, a global product-unitary equality implies equality of the pure-state factors up to component phases, for example by taking reduced states on each component. Conversely, tensoring component transformations supplies a transformation on the whole graph, with their phases absorbed into the global phase.
\end{proof}

\section{Cost of standardization}\label{app:standardcost}

\begin{lemma}\label{lem:standardcost}
On two connected graphs of order $n_c$, the common-standard-form procedure of \cite[Appendix~C]{CP} can be implemented in
\[
 \widetilde O\!\left(n_c^4 T_{\mathrm{rank}}(n_c)+n_c^6\right)
\]
bit operations, where $T_{\mathrm{rank}}(n_c)$ is the cost of binary matrix rank. Each returned ordinary-LC path has length $O(n_c^2)$.
\end{lemma}
\begin{proof}
The initial cover algorithm adds a minimal local set containing a previously uncovered vertex, so it retains at most $n_c$ sets and makes $O(n_c^4)$ cut-rank evaluations \cite[Section~3.2 and Algorithm~1]{MLS}. We account for the subsequent work using classical cubic elimination.

For a vertex set $L$, write $\operatorname{cutrk}_G(L)=\operatorname{rank}_{\F_2}\Gamma_{V\setminus L,L}$ for its cut rank. The set $L$ is a minimal local set exactly when $\operatorname{cutrk}_G(L)<|L|$ and $\operatorname{cutrk}_G(L\setminus\{v\})=|L|-1$ for every $v\in L$ \cite[Proposition~2]{MLS}. Hence its minimality and dimension can be checked with $O(n_c)$ rank computations, in $O(n_c^4)$ bit operations. A minimal local set inside any non-full-cut-rank set can be found by greedily deleting vertices while the remaining set is not full cut-rank. One pass suffices: every subset of a full-cut-rank set is full cut-rank. This also costs $O(n_c^4)$.

The termination arguments in \cite[Appendix~C]{CP} bound the total number of graph operations, including restarts. For $i=1,2$, let $X_i,Y_i$ be the current type-$\mathsf X$ and type-$\mathsf Y$ vertex sets in the $i$th graph. The potential $\sum_{i=1}^2(2|Y_i|+|X_i|)$ strictly decreases at steps~2--5. Step~6 preserves $|X_1|+|X_2|$, and a cover refinement at step~7 can only change a non-$\bot$ type to $\bot$. Thus there are $O(n_c)$ conflict-removing operations in total. Once type $\mathsf Y$ has disappeared, the sum of the vertex positions of all type-$\mathsf X$ vertices, in both graphs, is at most $n_c(n_c+1)$. It strictly decreases at each ordering pivot in step~6, and cannot increase at steps~5 or~7. There are therefore $O(n_c^2)$ such pivots over the entire run. Every refinement makes at least one further vertex type $\bot$, so there are $O(n_c)$ refinements. The final stars add at most $2n_c$ sets; the cover always has $O(n_c)$ members.

Each pass through step~7 inspects at most $2n_c$ stars. The $O(n_c)$ passes therefore require $O(n_c^2)$ minimality tests, costing $O(n_c^6)$ altogether. Checking the new sets in both graphs costs no more. For each admitted set, binary elimination computes its at most two independent stabilizer generators in $O(n_c^3)$ operations. Their Pauli labels determine the vertex types. Between cover additions, a local complementation merely permutes the three Pauli axes at its center and neighbors; these types can be updated directly, without recomputing the kernels. Updating adjacency matrices and searching for the next pivot cost $O(n_c^2)$ per operation, or $O(n_c^4)$ in total. The cover additions and all remaining bookkeeping fit within the $\widetilde O(n_c^6)$ bound. A pivot is three ordinary local complementations, giving the stated path-length bound.
\end{proof}

\section{Computational checks}\label{app:checks}

The modular compressor was compared with complete kernel enumeration on 288 small matrices and with the original incidence constraints on 1,536 binary incidence systems. Twelve larger Boolean-pattern controls require strict growth beyond the seed module. These tests include nonfree kernels and the boundaries $r=1$ and $t<3$.

For the determinant identity, all 12,886 linear solutions for the 1,461 ordered pairs of connected labelled graphs on two through four vertices were enumerated. Their determinant vectors were constant. The six-system test agreed with literal LC-orbit enumeration for every pair. Three choices of additional affine constraints per pair gave 4,383 comparisons, including constraints forcing rejection. The family includes 26 pairs for which every vertex in both graphs has odd degree. A disconnected two-edge example has determinant vector $(1,1,0,0)$, confirming why a single fixed vertex cannot impose invertibility across separate components.

The direct graph system was checked on 789 supplied-partition instances by enumerating all 2,717 candidate edge changes and, for each, the graphs reachable by complementations on the allowed vertices. The exact feasible coefficients agreed with the affine projection in every case. There were 239 feasible and 550 infeasible instances. A realizing complementation sequence was recovered in each feasible instance, with the Clifford equations checked after every step. These test Theorem~\ref{thm:direct} on prescribed partitions.

A further 90 instances compose incidence compression, kernel and image recovery, multiplicity lifting, and complementation extraction. The resulting local unitaries were applied to state vectors on four through eight qubits. Including one ordinary-LC instance with $F=\varnothing$, all 91 comparisons agreed with the target up to global phase, with Euclidean error below $10^{-10}$. The state-vector comparisons use double-precision complex arithmetic; the preceding modular, graph, and witness checks use exact integer arithmetic.

As a control beyond ordinary LC, we use the 27-vertex pair of Tsimakuridze and G\"uhne \cite[Section~7, Figure~7]{TG}. There are six vertices in $Z$, and one vertex in $X$ adjacent to each four- or five-element subset of $Z$. The partner graph adds the clique on $Z$. The 21 stars at $X$ form a common MLS cover; minimality was checked in each graph, giving 42 star checks in total. The constructive algorithm rejects level one and accepts levels two and three. At level two, it computes an image of dimension 15 and recovers multiplicities realizing the clique. The least-level procedure returns two, and returns one when the graph is compared with itself.

For the quotient calculation, 135 exact LC comparisons test \eqref{eq:orbitcosets} on three supplied standard forms: the 27-vertex graph, a 32-vertex extension with an attached five-cycle, and two copies joined by a $Z$--$Z$ edge on 54 vertices. The projected dimensions $(\dim\Omega,\dim\Omega_{\LC})$ are $(15,14)$, $(15,14)$, and $(30,28)$, respectively, giving counts two, two, and four. Each supplied MLS star and its type were checked by exact cut ranks. These comparisons use the unrestricted LC test and share its binary primitives with the quotient calculation.

The 31-vertex example in Section~\ref{sec:orbitcount} checks the dimension-two case and the strict inclusion $\Omega_1\subsetneq\Omega_{\LC}$. Direct enumeration of supported Pauli labels verifies its 21 dimension-one minimal sets and one dimension-two minimal set. A separate implementation of literal graph complementation checks the three-step boundary toggle, all 15 elementary $Z$--$Z$ changes, and the clique change, giving 16 realizing-sequence checks on the full graph. They agree with the projected dimension 15.

\begin{table}[ht]
\centering\small
\begin{tabular}{>{\raggedright\arraybackslash}p{0.49\linewidth}r >{\raggedright\arraybackslash}p{0.28\linewidth}}
\toprule
Test family & Cases & Comparison\\
\midrule
Modular matrices & 288 & Complete kernels\\
Binary incidence systems & 1,536 & Original constraints\\
Strict-growth controls & 12 & Full and compressed modules\\
Connected graph pairs & 1,461 & Determinants and LC orbits\\
Affine-constraint variants & 4,383 & Full solution sets\\
Direct graph systems & 789 & 239 yes; 550 no\\
Constructive state-vector checks & 91 & Error below $10^{-10}$\\
Published 27-vertex pair & 1 & Minimum level two\\
LC quotient comparisons & 135 & Unrestricted LC decisions\\
Boundary sequence checks & 16 & Literal complementations\\
\bottomrule
\end{tabular}
\caption{Checks of the algebraic stages. The affine-constraint variants use the connected graph pairs; these rows are overlapping test families. State-vector comparisons are floating-point checks, while the other comparisons use exact arithmetic.}\label{tab:checks}
\end{table}

The reference implementation also runs the full decision procedure from raw adjacency matrices. It computes an MLS cover, standardizes both graphs, solves the modular and affine systems, and returns the graph-operation certificate of Proposition~\ref{prop:certificate}. Literal LC-orbit enumeration supplies the reference answers for all 4,165 ordered graph pairs on one through four vertices and 120 further pairs on five and six vertices. The decisions agree in every case. Thirty additional positive pairs on 7, 10, 15, 20, and 28 vertices are generated by ordinary LC walks.

For the 27-vertex control, we apply independent random LC sequences to the two members of the published pair, using lengths 0, 27, 54, and 108. All 24 tests accept LU equivalence, while an unrestricted LC test rejects each pair. Only the resulting adjacency matrices are supplied to the LU algorithm; neither the scrambling sequences nor a prescribed partition is used. Twelve further 27-vertex negative controls are independently separated by exact cut-rank witnesses and are all rejected.

The certificate checker uses a separate implementation of the modular row closure and imports neither the standardization code nor the affine solver. Its row reduction is compared with brute-force annihilators on 108 modular systems, and its compressed incidence conditions with literal constraints on 48 systems at levels one through four. All returned yes-certificates are checked. Changed targets and malformed or non-independent witnesses are rejected in 29 further tests. These are finite implementation checks; the general statements, including arbitrary-target coverage in Lemma~\ref{lem:fixedcover}, rely on the proofs. The full decision implementation shares its modular and binary primitives with the earlier algebraic tests, and uses classical rank computation rather than fast matrix multiplication.

\section{Extracting a complementation sequence}\label{app:clifford}

\begin{proof}[Proof of Lemma~\ref{lem:restrictedsequence}]
The coordinate updates follow the graph proof of the Clifford criterion in \cite[Appendix~A]{CP}. Complementing $w$ on the source graph replaces its adjacency matrix $\Gamma$ and the solution coordinates by
\begin{align}\label{eq:cliffordupdate}
 \alpha'_u&=\alpha_u+\mathbf1_{u=w}\gamma_w,&
 \beta'_u&=\beta_u+\mathbf1_{u=w}\delta_w,\\
 \gamma'_u&=\gamma_u+\Gamma_{uw}\alpha_u,&
 \delta'_u&=\delta_u+\Gamma_{uw}\beta_u.\notag
\end{align}
Here $\Gamma$ is the adjacency matrix before the step. These formulas preserve \eqref{eq:cliffordlinear}--\eqref{eq:clifforddet}. If $\beta_w=\delta_w=1$, the step removes $w$ from the support of $\beta$, leaving its other coordinates unchanged.

If $\beta$ is nonzero and no such vertex exists, then $\delta_u=0$ and $\gamma_u=1$ on its support. For any vertex $u$ in that support, the diagonal linear equation gives
\begin{equation}\label{eq:pivotexists}
 \sum_w\Gamma_{uw}\beta_w\Gamma'_{wu}=1.
\end{equation}
Thus some neighbor $w$ of $u$ also has $\beta_w=1$ and $\delta_w=0$. Applying $u,w,u$ removes both vertices from the support. Explicitly, their pairs $(\beta_u,\delta_u),(\beta_w,\delta_w)$ change as
\[
 (1,0),(1,0)\ \longmapsto\ (1,0),(1,1)
 \ \longmapsto\ (1,1),(0,1)
 \ \longmapsto\ (0,1),(0,1).
\]
The edge $uw$ persists during these three operations, and no other $\beta$ coordinate changes.

Each one-step reduction removes one support vertex, and each three-step reduction removes two. Only current support vertices are used, so all centers lie in $W$ and the length is at most $\lfloor3|W|/2\rfloor$. When $\beta=0$, the determinant equations force $\alpha=\delta=1$, and the diagonal linear equation gives $\gamma=0$. The remaining equation is $\Gamma=\Gamma'$. Each complementation and coordinate update costs $O(n^2)$, giving total cost $O(n^3)$.

Conversely, start with the identity solution on $H$ and reverse any given complementation sequence from $G$ to $H$. Formula~\eqref{eq:cliffordupdate} changes a $\beta$ coordinate only at the center. A sequence centered on $W$ therefore gives a solution supported there. This proves the reverse implication.
\end{proof}

\end{document}